\documentclass[11pt]{article}
\pdfoutput=1 
\usepackage{microtype}
\usepackage{color}
\usepackage[usenames,dvipsnames,svgnames,table]{xcolor}

\usepackage{thm-restate}

\usepackage{mathrsfs}
\usepackage{amssymb,latexsym,amsmath}
\usepackage{amsthm, mathrsfs}
\usepackage{hyperref}
\usepackage{caption}
\usepackage{subcaption}

\usepackage{fullpage}
\usepackage{epsfig}

\theoremstyle{plain}
\newtheorem{definition}{Definition}
\newtheorem{theorem}{Theorem}
\newtheorem{problem}{Problem}

\newtheorem{lemma}{Lemma}

\newtheorem{observation}{Observation}
\newtheorem{question}{Question}

\usepackage{yfonts}
\usepackage{bm}

\usepackage{soul}

\usepackage{amsmath,amsfonts,amssymb,amsthm}
\usepackage{xspace}
\usepackage{paralist}

\newcommand{\ignore}[1]{} 
\newcommand{\comment}[1]{}

\newcommand{\spc}{\mathscr{S}}
\newcommand{\M}{\mathscr{M}}

\newcommand{\I}{\mathcal{I}}

\newcommand{\A}{\mathcal{A}}

\renewcommand{\S}{\mathcal{S}}
\renewcommand{\P}{\mathscr{P}}
\newcommand{\Q}{\mathcal{Q}}
\newcommand{\tp}[1]{\mathbf{p}^{(#1)}}
\newcommand{\ts}[1]{\mathbf{s}^{(#1)}}

\newcommand{\tw}[1]{w^{(#1)}}

\newcommand{\ta}[1]{\alpha^{(#1)}}
\newcommand{\br}{{\boldsymbol{\psi}}}
\newcommand{\rr}{\boldsymbol{\psi}}

\newcommand{\bp}{{\bf{p}}}
\newcommand{\bs}{{\bf{s}}}

\newcommand{\Pre}[1]{\ensuremath{P_{#1}}}

\newcommand{\allbad}{\#_{\mbox{\tiny{bad}}}}

\newcommand{\C}{{\mathscr{C}}}

\newcommand{\pl}{ {\mbox{polylog} }}
\newcommand{\poly}{ {\mbox{poly} }}

\renewcommand{\Box}{ {\mbox{Box\/} }}

\newcommand{\E}{\mathbb{E}}

\newcommand{\R}{\mathbb{R}}

\DeclareMathAlphabet{\mathpzc}{OT1}{pzc}{m}{it}

\renewcommand{\O}{\mathcal{O}}

\newcommand{\vol}{\mbox{Vol\/}}

\newcommand{\disk}[2]{\operatorname{Disk}(#1,#2)}
\newcommand{\dual}[1]{\operatorname{dual}(#1)}

\def\T{\mathcal{T}}
\def\D{\mathcal{D}}

\newcommand{\Frechet}{Fr\'echet\xspace}
\newcommand{\distFr}[2]{\ensuremath{d_F\pth{#1,#2}}}

\providecommand{\eps}{{\varepsilon}}%

\providecommand{\pth}[2][\!]{#1\left({#2}\right)}

\renewcommand{\Re}{\mathbb{R}}

\newcommand{\etal}{\textit{e{}t~a{}l.}\xspace}

\newcommand{\seclab}[1]{\label{sec:#1}}
\newcommand{\secref}[1]{Section~\ref{sec:#1}}

\newcommand{\lemlab}[1]{{\label{lem:#1}}}
\newcommand{\lemref}[1]{Lemma~\ref{lem:#1}}

\newcommand{\figlab}[1]{\label{fig:#1}}
\newcommand{\figref}[1]{Figure~\ref{fig:#1}}

\newcommand{\obslab}[1]{\label{obs:#1}}
\newcommand{\obsref}[1]{Observation~\ref{obs:#1}}

\title{On the complexity of range searching among curves}

\begin{document}
\author{Peyman Afshani \and Anne Driemel}
\date{}
\maketitle
\begin{abstract}
Modern tracking technology has made the collection of large numbers of densely sampled trajectories of moving objects widely available.  We consider a fundamental problem encountered when analysing such data: Given $n$ polygonal curves $S$ in $\Re^d$, preprocess $S$ into a data structure that answers queries with a query curve $q$ and radius $\rho$ for the curves of $S$ that have \Frechet distance at most $\rho$ to $q$.  

We initiate  a comprehensive analysis of the space/query-time trade-off for this data structuring problem.  Our lower bounds imply that any data structure in the pointer model model that achieves $Q(n) + O(k)$ query time, where $k$ is the output size, has to use roughly $\Omega\pth{(n/Q(n))^2}$ space in the worst case, even if queries are mere points (for the discrete \Frechet distance) or line segments (for the continuous \Frechet distance).  More importantly, we show that more complex queries and input curves lead to additional logarithmic factors in the lower bound.  Roughly speaking, the number of logarithmic factors added is linear in the number of edges added to the query and input curve complexity. This means that the space/query time trade-off worsens by an {\em exponential} factor of input and query complexity.  This behaviour addresses an open question (see~\cite{a12,Chan.ParTree}) in the range searching literature concerning multilevel partition trees which may be of independent interest, namely, whether it is possible to avoid the additional logarithmic factors in the space and query time of a multilevel partition tree.  We answer this question negatively.

On the positive side, we show we can build data structures for the \Frechet distance by using semialgebraic range searching.  The space/query-time trade-off of our data structure for the discrete \Frechet distance  is in line with the lower bound, as the number of levels in the data structure is $O(t)$,  where $t$ denotes the maximal number of vertices of a curve. For the continuous \Frechet distance, the number of levels increases to $O(t^2)$.

\end{abstract}
\thispagestyle{empty}

\setcounter{page}{0}
\pagebreak

\section{Introduction}

Recent technological advances have made it possible to collect trajectories of
moving objects, indoors and outdoors, on a large scale using various
technologies, such as GPS~\cite{Laurila_MDC_2012}, WLAN, Bluetooth,
RFID~\cite{rfid-2016} or video analysis~\cite{stats-2017}.  In this paper we
study time-space trade-offs for data structures that store trajectory data and
support similarity retrieval.
In particular we focus on the case where the similarity or distance between two
curves is measured by the \Frechet distance. This distance measure is widely
studied in computational geometry and gives high-quality results for trajectory
data.  We focus on the case where the query should return all input curves in a
specified \emph{range}, given by a query curve $q$ and a radius $\rho$.  The range
is defined as the set of curves that have \Frechet distance at most $\rho$ to $q$,
i.e., the metric ball of radius $\rho$ centered at $q$.  Our study is timely
as it coincides with the 6th GIS-focussed algorithm competition hosted by ACM
SIGSPATIAL\footnote{ 6th ACM SIGSPATIAL GISCUP 2017, see also
\url{http://sigspatial2017.sigspatial.org/giscup2017/}} drawing attention to this 
very problem from the practical domain. 

At the same time, our results address a broader question
concerning multilevel partition trees, a very important classical tool from the
range searching literature. See the following survey for more
background~\cite{Agarwal.Erickson.survey98}, but briefly, in range searching the
goal is to store a set of $n$ points such that the points in a query region can be
found efficiently. 
One of the most prominent problems is when the queries are simplicies in $\R^d$,
a problem known as {\em simplex range searching\/}.
Interestingly, the known solutions for simplex range searching can be easily 
repackaged into multilevel data structures that can even solve more difficult problems,
such as simplex-simplex searching: store a set of $n$ simplicies such that the simplicies
that are entirely contained in a query simplex can be found efficiently. 
For some illustrative examples on the versatility and power of multilevel data structures
see~\cite{Chan.ParTree}.

The concept of multilevel partition tree based data structures is broad and
mathematically not well-defined. 
Roughly speaking, in a multilevel data structure, first a base data structure
is built that defines some notion of first generation canonical sets. 
Then on the first generation canonical sets, a secondary set of data structures
are built which in turn defines a second generation canonical sets. 
Continuing this ``nested'' structure  for $t$ levels would yield a multilevel data structure
with $t$ levels.
This flexibility, allows more complex problems (such as simplex-simplex searching problem mentioned
above) to be solved and with very little effort and by only degrading space or query time by small factors.
It seems intuitively obvious that each additional level  should blow up the space or the query time of the
data structure and in fact all known data structures suffer an exponential factor in $t$ (often a $\log^{O(t)} n$ factor).
It has been conjectured that this should be the case but not even a polynomial blow up was proven before
(see~\cite{a12,Chan.ParTree}).

\paragraph{Exponential vs polynomial dependency.}
To better understand the situation, let us momentarily
focus on planar data structures with linear or near-linear space. 
For the main problem of simplex range reporting, there exist data structures with
$O(n)$ space and $O(n^{1-1/d}+k)$ query time where $k$ is the output size. 
This query time is conjectured to be optimal and there exist lower bounds that almost match it
up to a $2^{O(\sqrt{\log n})}$ factor (\cite{a12}) or $\log n$ factor (\cite{Chazelle.LB.JAMath89}). 
Thus, the base problem of simplex range reporting is well-understood.
However, beyond this, things are less clear. 
In particular, we would like to know what happens if the query regions or the input are more 
complex objects.
Assume the input is a set of $n$ points but the query is a tuple of $t$ hyperplanes that define a 
polytope.
To report the set of points inside the query polytope, 
we can triangulate the query polytope into $O(t^d)$ simplices
and then we can ask a simplex range searching query for each resulting simplex. 
This will not alter the space consumption at all and it will only blow up the
query time by a factor $O(t^d)$ but for a constant $d$,
this factor is a fixed {\em polynomial} of the query complexity, $t$.
This example shows that such ``obviously more complex'' queries can actually be handled very efficiently.
Now consider what happens if both queries and input are complex objects.
Consider a problem in which the input is a set containing $n$ tuples where the
$i$-th tuple $\bp_i$ is composed of 
$t$ points, i.e., $\bp_i = (\tp{1}_i, \dots, \tp{t}_i)$ where each $\tp{j}_i \in \R^d$,
and the query is a tuple
of $t$ simplices $(\sigma_1, \dots, \sigma_t)$.  
The goal is to report all the
input tuples $\bp_i$ such that $\tp{j}_i$ lies inside $\sigma_j$ for every $1 \le j \le t$.  
In this case, seemingly, the best thing to do is to build a multilevel
data structure composed of $t$ levels.  
Such a data structure will consume $O(n\log^{t-1}n)$ space and will have the
query time of $O(n^{1-1/d}\log^{t-1}n)$ using the best known results in the 
literature on multilevel data structures~\cite{Chan.ParTree}.
The crucial difference here is that both space and query time degrade
\emph{exponentially} in $t$ as opposed to the polynomial dependency in the previous case.
The main open question here is whether this exponential factor is required. 

The picture becomes more interesting once one looks at the history of multilevel data structures
and once one realizes that there are many ways to build them.
In Matou\v sek's~\cite{Matousek.RS.hierarchal.93} original paper, one would
sacrifice a $\log^2 n$ factor for space
and a $\log n$ factor for the query time but this comes at a larger pre-processing time.
If one wishes to reduce the preprocessing time, then the loss increases to an
unspecified number of $\log n$ factors per level.
Chan~\cite{Chan.ParTree} offers the currently best known way to build multilevel data structures
at only one $\log n$ factor loss for space and query time per level (in fact, in some cases, we can
do even better).
This brings us to the main lower bound question regarding multilevel data structures.

\begin{question}
Is it possible to avoid the additional logarithmic factors for every level in the space and query time of a multilevel data structure?
\end{question}

We at least partially settle this open question by showing that the space/query time tradeoff must blow up by at least a
roughly $\log n$ factor for every increase in $t$.
To do that, we show that a particular problem that can be solved using
multilevel data structures is hard.

We remark that the above question is ambiguous since we did not provide a
mathematically precise definition of a ``multilevel data structure''. Such a
definition would have to capture the versatility of the multilevel approach to
data structuring. For instance, multilevel partition trees can have different
fan-outs at different levels, they can selectively use duality restricted to
individual levels, or they can use different auxiliary data structures mixed in
with them.  A crucial and arguably most useful property of the multilevel
structures is that different levels can handle completely independent
subproblems. By lack of a precise definition that is commonly agreed upon and
perhaps in the hope to prove an even stronger statement, we take a different
approach: We prove a lower bound for a concrete relevant multilevel data
structuring problem (Problem~\ref{prob:mlsp} on page \pageref{prob:mlsp}).

The problem only involves independent points and simplicies (the basic
components of a simplex range reporting problem) and thus any multilevel data
structure must be able to solve the  problem.  This means, a lower bound for
this  problem gives a lower bound for the general class of multilevel data
structures.  From this point of view, our lower bound is in fact stronger:
it shows that the multilevel stabbing  problem is strictly more difficult than
the ordinary simplex range searching problem, even if we are not restricted to
use ``multilevel data structures.''

Not only that, but we also show that the lower bound also generalizes to geometric search
structures based on \Frechet distance: preprocess a set of $n$ polygonal curves of complexity
$t$ such that given a query polygonal curve of complexity $t$, we can find all input curves
within some distance $\rho$ of the query (Problems~\ref{prob:dds} and
$\ref{prob:fds}$ on page \pageref{prob:dds}).  This lower bound is not obvious
and it also provides additional motivation to study multilevel data structures.
The fact that we can extend our lower bound to such a practically relevant
problem emphasizes the relevance of our lower bounds.

\section{Definitions and Problem Statement}
A polygonal chain $s$ is a sequence of vertices $s_1,\dots,s_t \in \Re^d$. The
discrete \Frechet distance of two chains $s$ and $q$ is defined using the concept of traversals.
A \emph{traversal} is a sequence of pairs of indices $\{(i_1,j_1), (i_2,j_2), \dots, (i_k,j_k)\}$ 
such that $i_1=1, j_1=1, i_k=t_s$, and $j_k=t_q$ and one of the following holds
for each pair $(i_{m}, j_{m})$ with $m>1$:
\begin{inparaenum}[(i)]
\item $i_{m} = i_{m-1}$ and $j_m=j_{m-1}+1$, or
\item $i_{m} = i_{m-1}+1$ and $j_m=j_{m-1}+1$, or
\item $i_{m} = i_{m-1}+1$ and $j_m=j_{m-1}$.
\end{inparaenum}
The discrete \Frechet distance is defined as 
\begin{equation}
\distFr{s}{q} = \min_{T \in \T} \quad \max_{(i,j) \in T} \quad \| s_{i} - q_{j}\|  
\label{eq:discrete-frechet-def}
\end{equation}

Finding the traversal that minimizes the \Frechet distance is called the \emph{alignment problem}.

\medskip 
The continuous \Frechet distance is defined for continuous curves.
For a polygonal chain $s$, we obtain a polygonal curve by linearly interpolating $s_i$ and $s_{i+1}$, 
i.e., adding the edge $\overline{s_i s_{i+1}} = \{\alpha s_i + (1-\alpha)s_{i+1} ~|~ \alpha \in [0,1]\} $ 
in between each pair of consecutive vertices. 
The curve $s$ has a uniform parametrization that allows us to view it as a
parametrized curve $s: [0,1] \rightarrow \Re^d$.  The \Frechet distance between
two such parametrized curves is defined as 
\begin{equation} 
\distFr{s}{q} = \min_{f: [0,1] \rightarrow [0,1]} \quad \max_{\alpha \in [0,1]} \quad \| s(\alpha) - q(f(\alpha))\|,  
\label{eq:frechet-def}
\end{equation}
where $f$ ranges over all continuous and monotone bijections with $f(0)=0$ and $f(1)=1$.

\medskip
In this paper we consider the following two problems based on the \Frechet distance.

\begin{problem}[Discrete Frechet Queries]\label{prob:dds} Let $S$ be an input
set of $n$ polygonal chains in $\R^d$ where each polygonal chain
has size at most $t_s$.  Given a parameter $\rho$, we would like to store $S$ in a
data structure such that given a query polygonal chain $q$ of length at most
$t_q$, we can find all the chains in $S$ that are within the discrete Frechet
distance $\rho$ of $q$, see Equation \eqref{eq:discrete-frechet-def}.
\end{problem}

\begin{problem}[Continuous Frechet Queries]\label{prob:fds} Let $S$ be an input
set of $n$ polygonal curves in $\R^d$  where each polygonal curve 
consists of at most $t_s$ vertices.  Given a parameter $\rho$, we would like to store $S$ in a
data structure such that given a query polygonal curve $q$ consisting of 
$t_q$ vertices, we can find all the curves in $S$ that are within the continuous Frechet
distance $\rho$ of $q$, see Equation \eqref{eq:frechet-def}.
\end{problem}

For both problems, the output size is the number of input curves that match the query
requirements. 

Since we will be working with tuples of points and geometric ranges, we introduce the following
notations to simplify the description of our results. 

  A \emph{$t$-point} $\bf{p}$ in $\R^d$  is a tuple of $t$ points 
  $(\tp{1},\tp{2}, \dots, \tp{t})$ where each $\tp{i}$ is a point in
  $\R^d$.
  The concepts of $t$-hyperplanes and $t$-halfspaces, and etc. are defined similarly.
    A slab $s$ is the region between two parallel hyperplanes.
    The {\em thickness} of $s$ is the distance between the
    hyperplanes and it is denoted by $\tau(s)$.
  A \emph{$t$-slab} $\bf{s}$ in $\R^d$ is a tuple of $t$ slabs 
  $(\ts{1},\ts{2}, \dots, \ts{t})$ where each $\ts{i}$ is a slab. 
  The {\em thickness} of $\bs$ is defined as
  $\prod_{j=1}^t \tau(\ts{j})$.
  A $t$-point $\bf{p}$ is inside a $t$-slab $\bf{s}$ if
  $\tp{i}$ is inside $\ts{i}$ for every $1 \le i \le t$.
  We will adopt the convention that the $i$-th point $\bp$ in a $t$-point
  is  denoted by $\tp{i}$.
  The same applies to the other definitions.

We will also show a lower bound for the following concrete problem.
\begin{problem}[Multilevel Stabbing Problem]
\label{prob:mlsp}
  Let $S$ be an input set containing $n$ $t$-slabs.
  We would like to store $S$ in a data structure such that given a query $t$-point
  $\bp$ we can find all the $t$-slabs $\bs \in S$ 
  such that $\bs$ contains $\bp$.
\end{problem}

\paragraph{The pointer machine model.}
The model of computation that we use for our lower bound is the pointer machine model.
This model is very suitable for proving lower bounds for range reporting problems. 
Consider an abstract data structure problem where the input is a set $S$ of elements and where
a query $q$ (implicitly) specifies a subset $S_q \subset S$ that needs to be output by the data structure.
In the pointer machine model, the storage of the data structure is represented using 
a directed graph $M$ with constant out-degree where each vertex in $M$ corresponds to one memory cell.
Each memory cell can store one element of $S$.
The constant out-degree requirement means that each memory cell can point to at most a constant 
number of other memory cells.  
The elements of $S$ are assumed to be atomic, meaning,
to answer a query $q$, for each element $v \in S_q$, the data structure must visit
at least one vertex (i.e., cell) that stores $v$.
To visit that subset of cells, the data structure starts from a special vertex of $M$ (called the root) and 
follows pointers: the data structure can visit a memory cell $u$ only if it has already visited a cell
$v$ such that $v$ points to $u$. 
There is no other restriction on the data structure, i.e., it can have unlimited information and
computational power. 
The size of the graph $M$ lower bounds the storage usage of the data structure 
and the number of nodes visited while answering a query lower bounds the query time of the data structure.
Thus, when proving lower bounds in the pointer machine model, it is sufficient to show that if the data
structure operates on a small graph $M$, then during the query time, it has to visit a lot of cells
(or vice versa).

\section{Related Work on \Frechet Queries}
Few data structures are known which support \Frechet queries of some type.
We review the space and query time obtained by these data structures.
In the following, let $n$ denote the number of curves in the data structure
and let $t$ denote the maximum number of vertices of each curve.
The data structures can be distinguished by the type of queries answered:
\begin{inparaenum}[(i)]
\item nearest neighbor queries \cite{i-approxnn-02, ds-lshf-17},
\item range counting queries \cite{de2013fast, gs-fqt-15},
\end{inparaenum}

Before we discuss these data structures, we would like to point out that under
certain complexity-theoretic assumptions  both (i) and (ii) above become much
harder for long curves, and in particular for $t \in \omega\pth{\log n}$.  More
specifically, there is a known reduction from the orthogonal vectors problem
which implies that, unless the orthogonal vectors hypothesis fails, there
exists no data structure for range searching or nearest neighbor searching
under the (discrete or continuous) \Frechet distance that can be built in
$O\pth{n^{2-\eps}\poly(t)}$ time and achieves query time in $O\pth{n^{1-\eps}
\poly(t)}$ for any $\eps>0$ (see also the discussion in \cite{ds-lshf-17}).

A data structure by Indyk supports approximate nearest-neighbor searching under
the discrete \Frechet distance~\cite{i-approxnn-02}. The query time is in
$O\pth{t^{O(1)}\log n}$ and the approximation factor is in $O(\log t + \log\log
n)$. The data structure uses
space in $O\pth{|X|^{\sqrt{t}} (t^{\sqrt{t}} n)^2}$, where $|X|$ is the size of
the domain on which the curves are defined.  The data structure precomputes all
answers to queries with curves of length $\sqrt{t}$, leading to a very high
space consumption.   A recent result by
Driemel and Silvestri~\cite{ds-lshf-17} shows that it is possible to construct
locality-sensitive hashing schemes for the discrete \Frechet distance.  One
of the main consequences is a $O(t)$-approximate near-neighbor data structure
that achieves $O(t \log n)$ query time and $O(n \log n + tn)$  space. 

As for the continuous \Frechet distance, de Berg, Gudmundsson and Cook study
the problem of preprocessing a single polygonal curve into a data structure to
support range counting queries among the subcurves of this
curve~\cite{de2013fast}.  The data structure uses a multilevel partition tree to store
compressed subcurves. This representation incurs an approximation factor of
$2+3\sqrt{2}$ in the query radius.  For any parameter $n\leq s\leq n^2$,
the space used by the data structure is in $O\pth{s~\pl(n)}$. The queries are
computed in time in $O\pth{\frac{n}{\sqrt{s}}~\pl(n) }$.  However, the data
structure does not support more complex query curves than line segments.

The motivation to study the subcurves of a single curve originated from the
application of analyzing single trajectories of individual team sports players 
during the course of an entire game. A different application, namely   the map
matching of trajectories onto road maps~\cite{giscup2012} led Gudmundsson and
Smid to study slightly more general input---consider the geometric graph that
represents a road map and a range query among the set of paths in the graph.
Gudmundsson and Smid study the case where the input belongs to a certain class
of geometric trees~\cite{gs-fqt-15}. Based on the result of de Berg,
Gudmundsson and Cook they describe a data structure which supports approximate
range emptiness queries and can report a witness path if the range is
non-empty. Furthermore, the queries can be more complex than mere line
segments.  The data structure has size $O(n~\pl(n))$, preprocessing time in
$O(n~\pl(n))$ and answers queries with polygonal curves of $t$ vertices in
$O(t~\pl(n))$ time.  

It should be noted that the latter two data structures~\cite{de2013fast,
gs-fqt-15} make strict assumptions on the length of the edges of the query
curves with respect to the query radius which
seems to simplify the problem. While it is widely believed, based on
complexity-theoretic assumptions, that there is no $O(t^{2-\eps})$-time
algorithm for any $\eps>0$ that can decide if the discrete or continuous
\Frechet distance between two curves is at most a given value of
$\delta$~(see Bringmann~\cite{b-seth-14}), this problem drastically simplifies 
if $\delta$ is smaller
than half of the maximal length of an edge of the two curves. In particular, a
simple linear scan can solve the decision problem in $O(t)$ time.
Our results do not make any assumptions on the length of the edges of the
curves or the distribution of the edges.

%

\section{Our Results}
We show the first upper and lower bounds for exact range searching under the discrete and
continuous \Frechet distance. 
Our lower bounds are in fact obtained for the multilevel stabbing problem
and it proves that the space $S(n)$ required for  answering the multilevel stabbing queries in $Q(n) + O(k)$ time
must obey
\begin{align}
  S(n) = \Omega(( \frac{n}{Q(n)} )^2) \cdot \frac{\left( \frac{\log(n/Q(n))}{\log\log n}\right)^{t-1}}{2^{O(2^t)}} \mbox{\; as well as\; }
  S(n) = \Omega\left(  \frac{n}{Q(n)}\right)^2 { \Theta\left( \frac{\log(n/Q(n))}{t^{3+o(1)}\log\log n}\right)^{t-1-o(1)}}.\label{eq:results}
\end{align}
Here $k$ is the size of the output and $t$ is the number of levels.
Based on what we have discussed, not only this proves the first separation between the simplex range reporting data structures
and multilevel data structures, but it also shows space should increase exponentially in $t$, as long as $t \le (\log n)^{1/3 - \varepsilon}$.

For the \Frechet distance queries, a set of $n$ polygonal curves in
$\R^d$ is given as input, where each input curve consist of at most $t_s$
vertices.  A query with a curve of $t_q$ vertices and query radius $\rho$
returns the set of input curves that have \Frechet distance at most $\rho$ to $q$.

\newcounter{r-counter}
\begin{compactenum}[(i)]
\item Assume there exists a data structure that achieves $Q(n)$ query time  and
uses $S(n)$ space in the pointer model. We show that the $S(n)$ must obey the same
lower bound in Eq.~\ref{eq:results}
where $t \leq \min\pth{t_s/4, t_q/2}$.
\setcounter{r-counter}{\value{enumi}}
\end{compactenum}

In addition, we show how to build multilevel partition trees for the discrete
and the continuous \Frechet distance using semi-algebraic range searching:
\begin{compactenum}[(i)]
\setcounter{enumi}{\value{r-counter}}
\item For the discrete \Frechet distance we descibe a data structure which uses space in
$\O\pth{n (\log\log n)^{t_s -1}}$ and achieves query time in $\O\pth{n^{1-1/d}
\cdot \log^{O(t_s)} n \cdot t_q^{\O(d)}}$, assuming $t_q=\log^{O(1)} n$.  

\item For the continuous \Frechet distance we describe a data structure for $d=2$
which uses space in $\O\pth{n (\log\log n)^{O(t_s^2)}}$ and achieves query
time in $\O\pth{\sqrt{n} \log^{O(t_s^2)} n }$, assuming $t_q=\log^{O(1)} n$.
For the second data structure, the query radius has to be known at preprocessing time.
\end{compactenum}

\section{Outlines of the Technical Proofs}

\subsection{Outline of the lower bounds}
We first prove lower bounds for the reporting variant of the multilevel stabbing problem
in the pointer machine model. 
By what we discussed, this gives  a lower bound for multilevel data structures.
Next, we build sets of input curves and query curves that show the same lower bounds can 
be realized under the \Frechet distance.
Before we sketch the lower bound construction, we say a few words about the
lower bound framework we use.

\subsubsection{The framework of the proofs}
\seclab{framework}
Our reporting lower bound uses a volume argument by Afshani~\cite{a12}. 
This argument can be used to
show lower bounds for stabbing reporting queries, i.e., the input is a set of ranges and
a query with a point returns all ranges that contain this point. 
Imagine, we want to answer any query in $Q(n) + O(k)$ time where $k$ is the size of the output.
In order to set up the volume argument we need to define a set of queries $\Q$ that has volume one
and a set of input ranges, such that
\begin{inparaenum}[(i)]
\item each query point is covered by sufficiently many ranges (by at least $Q(n)$ ranges), and 
\item the volume of the intersection of any subset of $\ell$ ranges is sufficiently small, i.e.,
  at most $v$.
\end{inparaenum}
Then, the framework shows that the space is asymptotically lower bounded by
$Q(n) v^{-1}2^{-O(\ell)}$.
The intuition of why the framework works is the following: the intersection of
some subset of ranges is the locus 
of (query) points that must output those particular subset of ranges.
Thus, if the intersection of every subset of $\ell$ ranges is small, 
then our set of queries $\Q$ contains many different queries that each output
a different subset of input ranges. Thus, precomputing (and implicitly
storing) partial answers must increase the space according to these volumes.

\subsubsection{Multilevel Stabbing Problem}

We start with the unit cube $\Q$ in $\R^{2t}$ where $t$ denotes the number of levels.
In particular, we view $\Q$ as the Cartesian product of $t$ unit squares:
$\Q  = \Q_1 \times \dots \times \Q_t$.
The input is a set of $n$ $t$-slabs in $\Q$.
The query is a $t$-point in $\Q$.
The main part of the proof is an intricate construction of the $t$-slabs.

The main result here is Lemma~\ref{lem:construction} (see page~\pageref{lem:construction}).
We will not repeat the exact technical claim 
and instead we will focus on the general ideas and the intuition behind them. 
The first step is to build $r=Q(n)$ different sets of $t$-slabs $\S_1, \dots, \S_r$ of 
roughly $n/r$ size, such that the slabs in each set $\S_i$ tile $\Q$, i.e., any 
$t$-point is covered by exactly one $t$-slab.
This will directly satisfy condition (i) of the framework described in \secref{framework}. 
The difficult part is to find a good construction that guarantees that every subset of 
$\ell$ $t$-slabs have a small intersection.

To build $\S_i$, 
we build a set $\S_{i,j}$ of  two-dimensional slabs in each unit square $\Q_j$ such that
they together tile $\Q_j$.
Then, $\S_i$ is taken to be the set of $t$-slabs that one obtains by
creating the Cartesian product of all the slabs created in $\Q_1, \dots, \Q_t$.
See Figure~\ref{fig:hyperslabs} on page \pageref{fig:hyperslabs} for an illustration.
In order to obtain small intersection volume we would like to adjust the
thickness of the two-dimensional slabs.  While adjusting the thickness of the
slabs in each universe, we make sure that we 
create roughly $n/r$ $t$-slabs in $\S_i$:
This boils down to making sure that the product of the thicknesses of the two-dimensional
slabs is a fixed value $\tau$.
We have $t$ degrees of freedom to pick the orientation of the slabs and thus we can
represent the set of angles that define the orientation of the slabs in each $\Q_j$ 
by a point in $\R^t$; we call these points, ``parametric points''.
Thus, every set $\S_i$ has one parametric point and in our construction there are
$r$ parametric points in total.

The parametric points need to be placed very carefully. 
In particular our construction places the parametric point such that the volume of the smallest
axis-aligned rectangle created by any two parametric points is maximized (see Lemma~\ref{lem:orthc}
in page \pageref{lem:orthc}).
Intuitively, this means that if the points are ``well-spread'' so that no small axis-aligned
box can contain two points, then the volume of the intersection of the slabs is also going to be large. 

Regarding the thicknesses of the slabs, 
we have $t-1$ degrees of freedom since the product of the thicknesses is set to be a fixed
value $\tau$. 
However, we need to place more restrictions on the thicknesses. 
We make sure that the different thicknesses are sufficiently different.
In particular, we set the width to be in the form of $R^w$ for some fixed parameter $R$ and some
integer $w$. 
This means that for each slab we allow roughly a logarithmic number of different possible thicknesses. 
Thus, the $t-1$ degrees of freedom in choosing the thicknesses are translated to 
freedom in choosing $t-1$ integers in some narrow range (between 0 and roughly $\log_R n$).
Note that this freedom is only present for the first $t-1$ two-dimensional slabs,
that is, in $\Q_1, \dots, \Q_{t-1}$ and the thickness of the last slab is determined
based these values and the value of $\tau$.
This further implies that the sum of the integers that we choose 
should also be within the same narrow range. 
Nonetheless, unlike the case of angles, our choices in picking these integers
are represented {\em combinatorially\/} as a single value and
we treat it like a color. 
In other words, we define a set of all the available colors (roughly $((\log_R n)/t)^{t-1})$
and then associate each set $\S_i$ with a color; the color determines the thickness of slabs in 
$\Q_1, \Q_2, \dots, \Q_t$.

Thus, after placing $r$ parametric points in $\R^t$, we need to color each point with a color.
This coloring needs to be done carefully as well. 
The placement and the coloring of the points are done using one lemma (Lemma~\ref{lem:orthc}
in page \pageref{lem:orthc}).

However, more work is required to make the construction work. 
We need to impose some favorable combinatorial structure on the set of colors that we create
by removing some of the colors. 
This is done by sampling a small number of colors. 

Finally, we try to bound the volume of the intersection of any $\ell$ of the $n$ $t$-slabs that we created.
Any two slabs in $\S_i$ are disjoint and thus for a non-zero intersection, the $\ell$ slabs should come
from $\ell$ different sets, e.g., $\S_1, \dots, \S_\ell$. 
The straightforward argument gives us a bound on $v$ that ultimately gives the same lower bound
as simplex range reporting.
So we perform a non-obvious analysis. 
We look at two possible cases:
Either (i) two of the parametric points of $\S_1, \dots, \S_\ell$ have the same color 
and in this case we  use the properties of our coloring (see~Lemma~\ref{lem:orthc} in page \pageref{lem:orthc})
to ensure that such points are ``well-spread''; in particular, if we have
$n_c$ colors, we can make sure that the parametric points of each color are a factor roughly
$n_c$ ``better spread'', meaning, the volume of the smallest axis-aligned rectangle that contains
two points of the same color is a factor $n_c$ larger than the volume of the smallest rectangle that
contains two  points of different colors. 
Ultimately, this buys us a $n_c$ factor in our lower bound. 
Observe that the value $n_c$ grows exponentially on $t$ (up to some maximum value).
The other case is when (ii) all the parametric points of $\S_1, \dots, \S_\ell$ have distinct colors.
By  using the favorable combinatorial property that we had imposed earlier on the set of colors,
we find 3 colors  among the many different distinct colors and an index $j$
such that these colors have three distinct values at coordinate $j$.
This in turn implies that the slabs in $\Q_j$ have 3 distinct thicknesses.
However, thicknesses differ by at least a factor $R$ and thus further analysis 
buys us a factor $R$ on value of $v$.
By combining the two cases, we show that we can improve our lower bound by either a 
factor $n_c$ or factor $R$.
We set our parameters such that $R$ and  $n_c$ are roughly equal and we obtain the lower bound. 

\subsubsection{Constructions for the \Frechet Distance}
In order to apply the above construction to the \Frechet range searching we 
dualize the \Frechet query ranges to some extent.  Our dualization differs
significantly between the two variants of the problem.  For the discrete
\Frechet distance we observe that the set of points that lie within \Frechet
distance $\rho$ to a line segment are contained in the intersection of the two circles of
radius $\rho$ centered at the two endpoints. 
We call the intersection of two circles a {\em lens}.
Thus, we create a set  of lenses as input instead of a set of slabs and we let
the vertices of the query curve act as stabbing queries. 
Refer to \figref{lb_discrete} on page \pageref{fig:lb_discrete} for an
illustration of this straight-forward approach.  We observe that inside the
unit square, lenses can be made  to almost look like slabs, that is, for any
slab, we can create a lens that is fully contained in the slab such that the
area of the symmetric difference between the slab and the lens is made
arbitrarily small.  As a result, after a little bit more work, we can show that
the construction for the multilevel stabbing problem directly gives a lower
bound for the discrete \Frechet queries problem.

In contrast, our construction for the continuous \Frechet distance dualizes the
\emph{lines} supporting the edges of the query curve, creating a separate ``universe'' for
every odd edge (in lieu of a universe for every vertex).  Here, our
construction is such that the locus of query curves in the dual space that
lie within \Frechet distance $\rho$ to a specific input curve forms a set of
slabs---one in each universe. To this end we let the input curve follow a
zig-zag shape. We use one zig-zag curve per universe. Refer to \figref{zigzag_gadget}
on page \pageref{fig:zigzag_gadget}
for an example of a zig-zag used in the construction.   Our analysis uses the
basic fact that the set of lines intersecting a vertical interval in the
primal space corresponds to the set of points enclosed in a slab in the dual
space. We combine this fact  with a well-known connection between the
\Frechet distance and ordered line-stabbing initially observed by Guibas
\etal~\cite{guibas1993approximating}. This observation says that  the line
supporting the query edge needs to stab the disks of radius $\rho$ centered at the
input curve in their correct order. For our zig-zag curves this has the effect
that the line needs to intersect the vertical interval formed by the two
intersection points of the circles of radius $\rho$ centered at the two corners of
the zig-zag. We can control the width, orientation and position of the
resulting slab in the dual space by varying the length and the position of this
vertical interval. Using these proof elements, we can show that the lower bound of
the multilevel stabbing problem which is analyzed in the beginning, carries
over to the continuous \Frechet distance as well. 

\begin{figure}\centering
\includegraphics{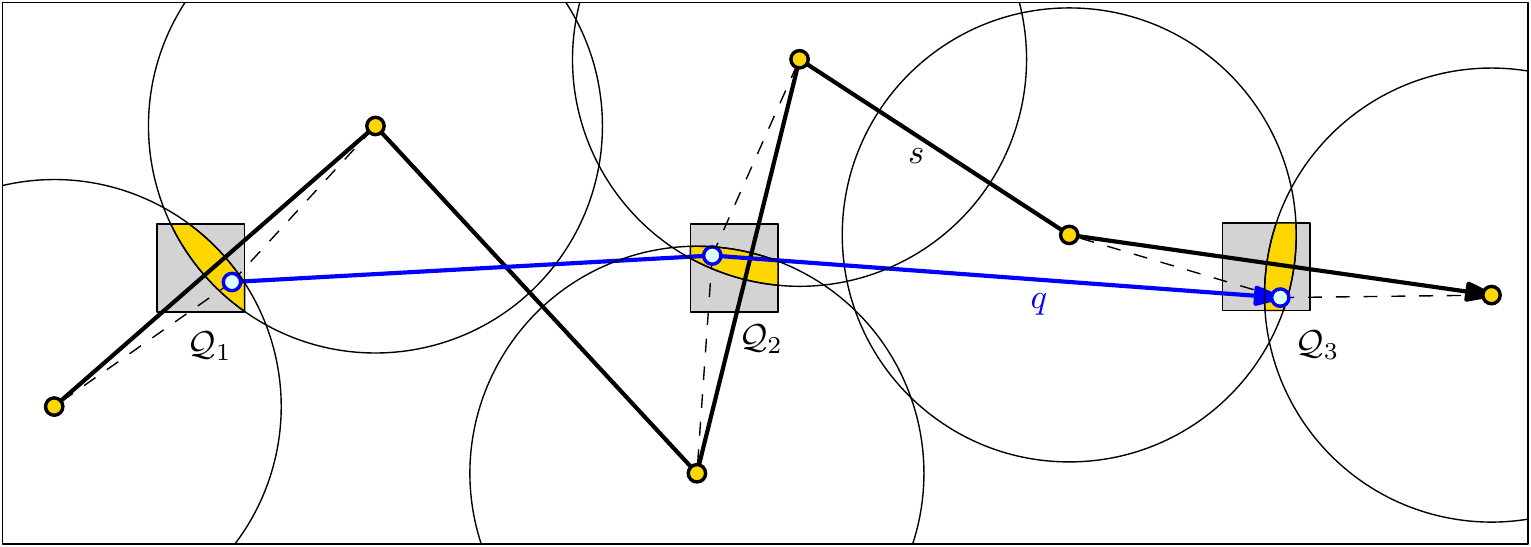}
\caption{Illustration of the lower bound construction for the discrete \Frechet
distance showing universes $\Q_1$, $\Q_2$ and $\Q_3$. For every $i$, a query curve $q$ has its $i$th
vertex inside $\Q_i$. The intersection of two disks centered at the vertices 
of the $i$th odd edge of an input curve $s$ forms a ``near-slab'' and needs to contain 
the $i$th vertex of a query curve if $s$ is contained in the query range centered at $q$.  } 
\figlab{lb_discrete}
\end{figure}

\subsection{Outline of the data structures}

To obtain our upper bounds, we perform an extensive analysis of the definition
of the \Frechet distance that allows us to restate the alignment problem using
a sequence of semialgebraic range queries.
One of the challenges here is to design a set of filters that do not create
duplicates in the output across the different range queries that need to be
performed. We first focus on the discrete \Frechet distance, where the analysis is
significantly cleaner and simpler. The dynamic programming algorithm which is commonly used to
compute the discrete \Frechet distance uses a Boolean matrix, the so-called
free space matrix, to decide which alignments between the curves are feasible.
The entry $(i,j)$ of this matrix indicates if the Euclidean distance between
the  $i$th vertex of one curve and the $j$th vertex of the other curve is at
most $\rho$. The two curves have \Frechet distance at most $\rho$ if and only if
there exists a traversal that only uses the 1-entries in the free-space matrix.
The set of possible truth assignments to this matrix induces a partition on the
input curves with respect to their free space matrix with the query curve.
Furthermore, each set in this partition is either completely contained in the
query range or it is completely disjoint from the query range.  We show how to
construct a multilevel data structure that allows us to query independently
for each of those sets which are contained in the query range. 

Our query processing works in three phases.  First, we compute all feasible
free space matrices based on the arrangement of balls centered at vertices of
the query. Next, we refine this arrangement to obtain cells of constant
complexity that can be described by the zero set of a polynomial function. In
the third phase we query the data structure with each free space matrix
separately, using semialgebraic range searching in each level of the data
structure to filter the input curves that have their $i$th vertex inside a
specific cell of the refined arrangement.  To see how this works, consider the
set of $i$th vertices of the input curves that lie in a fixed cell of the
arrangement of balls centered at the vertices of the query curve.  The
corresponding input curves share the same truth assignment in the $i$th column
of the free-space matrix with $q$. Refer to \figref{query_plan} for an example. 
\begin{figure}\centering
\includegraphics{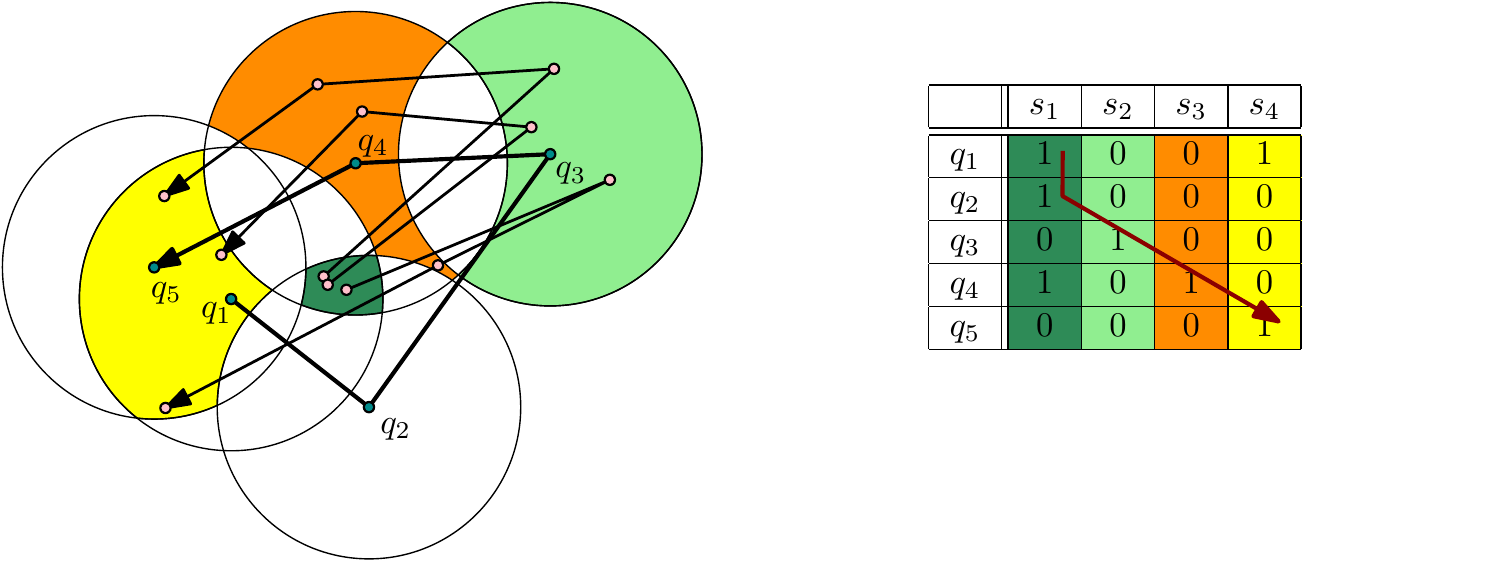}
\caption{Example of a query matrix for the discrete \Frechet distance with a
feasible traversal (right). The truth assignment in a fixed column corresponds
to a cell in the arrangement of balls centered at the vertices of the query
curve (left).  The figure also shows three input curves that have this
free-space matrix with the query curve and would thus be reported. } 
\figlab{query_plan}
\end{figure}

We now build a standard multilevel partition tree on the polygonal chains. In
the $i$th level we store the $i$th points of the input curves.  Our query
algorithm processes the free-space matrix in a column-by-column fashion, where we use
the convention that the column index refers to a point on the input curves and a row
index refers to a point on the query curve.  This makes the storage layout of the
data structure independent of the number of vertices of the query curve.

For the continuous \Frechet distance the approach is similar, at least on a
high level. The main difference is that the Boolean matrix that guides the queries
is more complicated, since we operate on the continuous free-space diagram
instead of the discrete free-space matrix. We first define high-level
predicates that carry enough information to decide the \Frechet distance. Each
predicate involves a constant number of edges and vertices from the input and
query curves, e.g., testing the feasibility of a monotone path for a
combination of a row and two vertical lines in the free-space diagram.  Next we
show how to represent these predicates using more basic operations, e.g.,
above-below relationships between points and lines that can be dualized.
Finally, the query algorithm will test groups of these predicates for each
feasible truth assignment separately.  
Also here we manage to keep the layout of the data structure independent of the
complexity of the query curve.

There are two main challenges in dealing with the continuous case. One is to 
obtain the more complicated discrete matrix that captures all possible free-space
diagrams of the fixed query curve with any {\em arbitrary} possible input curve. 
The second challenge is to make sure we can express all our predicates in the
framework of semialgebraic range searching in two dimensions.  Our solution is
non-obvious since the \Frechet distance is not defined as a closed-form algebraic
expression.  This second challenge is the main issue that prevents us from
directly generalizing our data structure to higher dimensional queries.

\subsection{Organization}

We prove the lower bounds in \secref{lower-bound}. 
We first show the lower bound for the multilevel stabbing problem.
The construction is given in \secref{lb-construction}. We discuss the range
reporting lower bound in \secref{lb-reporting}.
In \secref{lower-bound-frechet} we show how to
implement the construction for the two variants of \Frechet queries. 
We describe our data structures in \secref{data-structure}.  In
\secref{ml-data-structure} we describe the machinery that we use to build our
data structures.  In \secref{dfq} we develop a data structure for discrete
\Frechet queries. In \secref{cfq} we extend these ideas and develop a data
structure for continuous \Frechet queries.   We conclude with some open
problems in \secref{conclusions}.

\section{Lower Bounds}
\seclab{lower-bound}
As discussed, we prove lower bounds for a concrete {\em problem\/}, that is, the multilevel stabbing problem.
To do that, we need to construct a ``difficult'' input instance of 
$n$ $t$-slabs with certain desirable properties.
This construction is at the heart of our lower bounds and this is what we are going to attempt in this section.

\subsection{Definitions}
Our lower bounds for the multilevel stabbing problem is based on an intricate construction
that we outline in this subsection. 
Define the space $\Q = \Q_1 \times \Q_2 \times \dots \Q_t$ where each $\Q_i$ is
the unit cube in the plane.
$\Q$ now represents the set of all possible queries:
a query $t$-point $\bp$ is represented by the point $\bp \in \Q$ which corresponds to
a point $\tp{i}$ in $\Q_i$, for every $1 \le i \le t$.
Observe that the points $\tp{i}$ are completely independent. 
Similarly, an input $t$-slab $\bs$ is represented by picking $t$ independent slabs, one slab $\ts{i}$ in $\Q_i$
for each $1 \le i \le t$.

Consider a (measurable) subset $f \subset \R^{D'}$ that lies in a $D$-dimensional flat $V$
of $\R^{D'}$, $D \le D'$.
We denote the $D$-dimensional Lebesgue measure of $f$ with $\vol_D(f)$.
For a set of points $q_1, \dots, q_s \in \R^D$, we denote the smallest
axis-aligned box that contains  them all with $\Box_D(q_1, \dots, q_s)$.
Finally, for two $t$-slabs $\bs_1$ and $\bs_2$, we say $\bs_1$ is a {\em translation} of $\bs_2$
if for every index $j$, the slabs $\ts{j}_1$ and  $\ts{j}_2$ are parallel and have the same thickness.

\subsection{The 2D Construction} \seclab{lb-construction}
\begin{restatable}{lemma}{lemconstruction} \label{lem:construction}
    Consider parameters $t, \tau, r, R, n_c$ and $\ell$ under constraints to be specified shortly. 
    We can build a set $\S$ of $r$ $t$-slabs such that $\tau(\bs) = \tau$, for every $\bs \in \S$.
    Furthermore, 
    (i) for any $\ell$ $t$-slabs $\bs_1, \bs_2, \dots, \bs_\ell \in \S$ and any
    $\ell$ $t$-slabs $\bs'_1, \dots, \bs'_t$ such that $\bs'_j$ is a translation of
    $\bs_j$, we have $\vol_{2t}(\bs'_1 \cap \bs'_2 \cap \dots \cap \bs'_\ell) \le 
  \max\left\{{\frac{\tau^2 rt^{t+o(t)}}{n_c}, \frac{\tau^2 rt^{t+o(t)}}{R}} \right\}$.

  The constraints are that  $\Omega(1) \le r \le n/2$, $\Omega(1) \le R \le n^{1/(2t)}$, $2\le \ell <r$
  and that $n_c$ is defied as $n_c = (\frac{\log_R n}{t})^{t-1}$ when $\ell \ge 2^t$ and
  $n_c = \Theta( (\frac{\log_R n}{t})^{t-1}2^{-t}(\frac{\log_R n}{t})^{-t/\ell})$ when $\ell < 2^t$.
\end{restatable}

As we shall see later, combined with the existing framework, the above lemma 
offers our desired lower bounds with only little bit more work. 
Thus, the main challenge is actually proving the above lemma.
The main idea is the following:
To define each $t$-slab in $\S$, 
we have the freedom to pick $t$ different angles, one angle for each universe $\Q_j$, $1\le j \le t$.
We also have the freedom to alter the thickness of the slabs we constructed in each $\Q_j$. 
Thus, we have ``$t$ degrees of freedom'' to pick the angles and ``$t-1$ degrees of freedom'' to pick the
slab thickness. 
The former $t$ degrees of freedom are represented as points (that we call ``parametric points'')
in $\R^t$ and the latter are represented combinatorially as ``colors''.
To make the construction work, we do not allow for all possible combinations of ``colors'' and instead 
we prune the colors using a combinatorial technique.
We ultimately isolate a sub-problem that is very connected to orthogonal range searching.
This is a very satisfying since it was suspected that there could be connections between orthogonal
range searching and multilevel non-orthogonal range searching~\footnote{For example, Chan~\cite{Chan.ParTree}
compares non-orthogonal multilevel data structures to $d$-dimensional range trees that can be
viewed as $d$-levels of $1$-dimensional data structures.}.
As a result, we manage to incorporate some techniques from orthogonal range searching lower bound in our
construction (see Theorem~\ref{thm:points}).
However, combining the colors (i.e., the ``orthogonal component'') and the set of parametric points
(the non-orthogonal component) requires a careful analysis. 

\subsubsection{Parameters Defining the $t$-slabs.}
To construct each slab in $\S$,  we use $2t-1$ parameters: 
assume, we would like to construct a slab $\bs_i  \in \S$. 
We use $t$ real-valued parameters
$\alpha^{(1)}_i, \alpha^{(2)}_i, \dots, \alpha^{(t)}_i$ and $t-1$ integral
parameters $w^{(1)}_i, w^{(2)}_i, \dots, w^{(t-1)}_i$.  
We call these $2t-1$ parameters the {\em defining parameters of $\bs_i$}.
$\alpha^{(j)}_i$ is the angle slab $\ts{j}_i$ makes with the $X$-axis, and the thickness of
$\ts{j}$ is defined as $\tau(\ts{j}_i) = R^{-w^{(j)}_i}$, for $1 \le j \le t-1$.
However, since we would like to end up with a $t$-slab $\bs_i$ such that
$\tau(\bs_i) = \tau$, we define 
$w^{(t)}_i = \log_R(\tau^{-1}) - \sum_{j=1}^{t-1}w^{(j)}_i$ and $\tau(\ts{t}_i) = R^{-w^{(t)}_i}$.
Note that $w^{(t)}_i$ is not necessarily an integer.
We have:
\begin{observation}\label{ob:widthprod}
  \[ \tau(\bs_i) = \prod_{1\le j \le t} \tau(\ts{j}_i)=\tau(\ts{t}_i) \prod_{1\le j \le t-1} \tau(\ts{j}_i) = R^{-(\log_R(\tau^{-1}) - \sum_{j=1}^{t-1}w^{(j)}_i)  - \sum_{j=1}^{t-1}w^{(j)}_i}= \tau.\]
\end{observation}

\begin{definition}
  Consider a $t$-slab $\bs_i \in \S$. Let $\alpha^{(1)}_i, \dots, \alpha^{(t)}_i$ and 
  color $(w^{(1)}_i, \ldots, w^{(t)}_{i})$ be the defining parameters of $\bs_i$.
  We call the point $\phi(\bs_i) = (\alpha^{(1)}_i, \dots, \alpha^{(t)}_i)\in \R^t$ the parametric point of $\bs_i$
  and denote it with $\phi(\bs_i)$ and with 
  the tuple $(w^{(1)}_i, w^{(2)}_i, \dots, w^{(t-1)}_i)$ being its color.

\end{definition}

We have to make very careful choices when picking the defining parameters of
$\bs_i$. We discuss how to pick the colors in Section~\ref{sec:colors}.

\begin{figure}\centering
\includegraphics{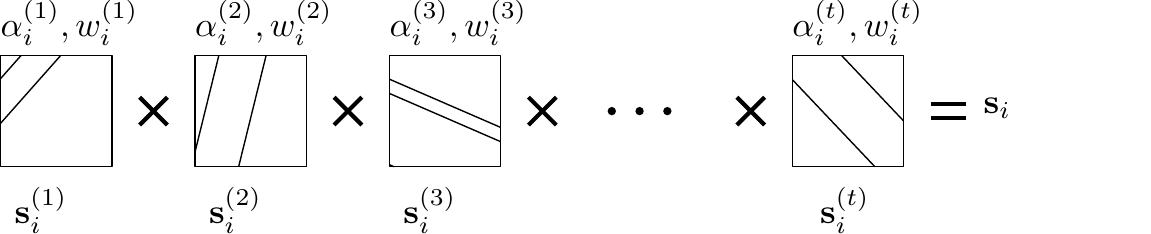}
\caption{Each $\bs_i$ is defined as the Cartesian product of $t$ 
  $2$-dimensional slabs. The thickness of $\ts{j}_i$ is $R^{-w^{(j)}_i}$.}
\label{fig:hyperslab}
\end{figure}

We now establish some basic facts about this construction.
\begin{observation}\label{ob:areapi}
  Consider $\ell$ $t$-slabs $\bs_1, \dots, \bs_\ell$.
  We have $\vol_{2t}(\bs_1 \cap \bs_2 \cap \dots \bs_\ell) = \prod_{i=1}^t \vol_2(\ts{i}_1\cap \ts{i}_2\cap \dots \cap \ts{i}_\ell)$.
\end{observation}

We will also use the following elementary geometry  observation regarding the area of 
the intersection of two slabs. 

\begin{observation}\label{ob:slabint}
  Consider two 2-dimensional slabs $s_1$ and $s_2$ of thickness $\tau(s_1)$ and $\tau(s_2)$ respectively.
  And let $\alpha$ be the angle between them. 
  Then, $\vol_2(s_1 \cap s_2) = O(\tau(s_1)\tau(s_2)/\alpha)$.
  (See Figure~\ref{fig:slabs}.)
\end{observation}

\begin{figure}\centering
  \includegraphics[scale=0.5]{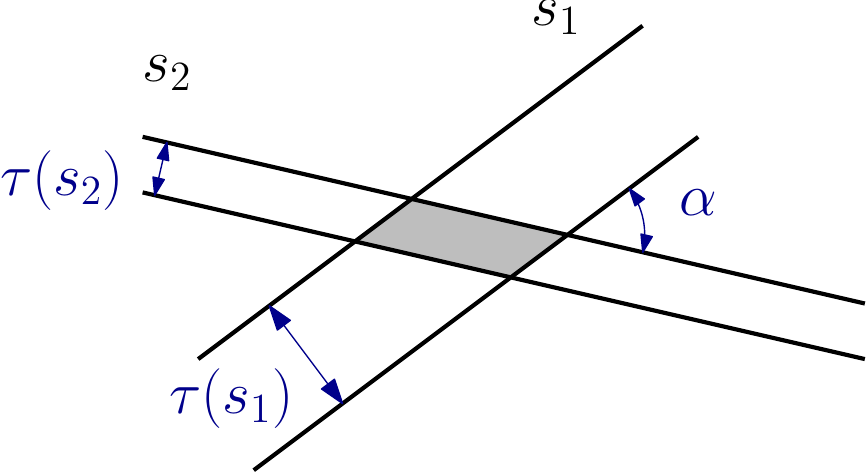}
\caption{Distance between the parallel lines forming slabs $s_1$ and $s_2$ is $\tau(s_1)$ and $\tau(s_2)$ respectively.
    If the angle between the slabs is $\alpha$, then the area of the shaded region is $\Theta(\tau(s_1)\tau(s_2)/\alpha)$.  }
\label{fig:slabs}
\end{figure}

\begin{lemma}\label{lem:twoint}
  Let $\bs_1, \bs_2 \in \S$ be two $t$-slabs and let $\phi(\bs_1)$ and $\phi(\bs_2)$ be the
  parametric points of $\bs_1$ and $\bs_2$ respectively. 
  Then regardless of the colors of these two $t$-slabs,  $\vol_{2t}(\bs_1 \cap \bs_2)$ is asymptotically bounded by 
  \[
      \frac{\tau^2}{\vol_t(\Box_t(\phi(s_1),\phi(s_2))}.
  \]
\end{lemma}
\begin{proof}
  By Observations~\ref{ob:slabint} and~\ref{ob:areapi},
$\vol_{2t}(s_1 \cap s_2)$ is asymptotically upper bounded by 
    \[
      \prod_{i=1}^t \frac{\tau(\ts{i}_1)\tau(\ts{i}_2)}{|\alpha^{(i)}_1 - \alpha^{(i)}_2|}.
    \]
 By Observation~\ref{ob:widthprod}, the nominator equals $\tau^2$ since the thickness of both $\bs_1$ and $\bs_2$ is $\tau$. 
 Then, the lemma then follows from observing that $\prod_{i=1}^t |\alpha^{(i)}_1 - \alpha^{(i)}_2|$ is
 exactly $\vol_t(\Box_t(\phi(s_1),\phi(s_2))$.
\end{proof}

\subsubsection{Coloring and the Parametric Points}
In this subsection, we will discuss how to pick the parametric points of the slabs in $\S$.
Essentially, we will place a set of $r$ points in $\R^t$ using the upcoming constructions. 
We extend the following construction that is used in lower bounds for the orthogonal problems.

\begin{theorem}\cite{aal09,aal10,Chazelle.LB.reporting}
  For any parameter $N$, we can place a set $P$ of $N$ points inside the unit cube in $\R^D$ such that 
  for any two points $p, q \in P$, we have $\vol_D(\Box_D(p,q)) =\Omega(1/N)$ where the constant
  in the asymptotic notation depend on $D$.
\end{theorem}

To choose the parametric points, we use the method in~\cite{aal09}.
But since in our case, the dimension $t$ is not longer considered a constant, we need to
provide a tight analysis, and determine its precise dependency
on the dimension (by using the prime number theorem).

\begin{theorem}\label{thm:points}
  For any parameter $N$, we can place a set $P$ of $N$ points inside the cube $[N]^D$ in $\R^D$ such that 
  for any two points $p, q \in P$ we have
  $\vol_D(\Box_D(p,q))=\Omega(\frac {N^{D-1}}{D^{D+o(D)}})$.
\end{theorem}
\begin{proof}
    We use the same construction in~\cite{aal09}:
    We pick the first $D-1$ prime numbers $a_1, \dots, a_{D-1}$ and we place $N$ points on the
    integer points in $[N-1]^D$.
    The coordinates of the $i$-th point $p_i$ is $p_i = (\overline{i_{(a_1)}}, \overline{i_{(a_2)}}, \dots,\overline{i_{(a_{D-1})}},i)$, 
    for $0 \le i < N$,
    where $\overline{i_{(x)}}$ is the ``reversed'' (or inverted) representation of $i$ in base $x$ using $\lfloor \log_{x}N\rfloor+1$ digits
    \footnote{That is, 
      if $i = c_0 + c_1 x + c_2 x^2 + \dots + c_{\lfloor \log_{x}N\rfloor+1} x^{\lfloor \log_{x}N\rfloor+1}$, then $\overline{i_{x}} = c_{\lfloor \log_{x}N\rfloor+1} + c_{\lfloor \log_{x}N\rfloor} x + \dots + c_0 x^{\lfloor \log_{x}N\rfloor+1}$.
    Basically, we write $i$ in base $x$ with the most significant digits to the left and then read the digits from right to left to obtain
    $\overline{i_x}$.}.
    In~\cite{aal09}, it is only proved that the volume of the axis-aligned box that contains $2$ points is 
    $\Omega(\frac 1{N})$ with a ``constant'' that depends on $D$.
    Here, we need to make the dependence on $D$ explicit.

    Consider two points $p_{k}, p_{k'}$ such that $k < k'$.
    Let $d = k' - k$.
    Observe that if $x^y$ divides $d$ but $x^{y+1}$ does not, then the representation of $d$ in base $x$
    contains exactly $y$ leading zeros (and as a conclusion, $\overline{d_{(x)}}$ contains exactly zeros at its
    $y$ most significant digits).
    This implies that for any natural number $z$ such that, $z +d < N$, $\overline{{z}_{(x)}}$ and $\overline{z+{d}_{(x)}}$ agree
    on exactly $y$ of their most significant digits, which yields the bound 
    $|\overline{z_{(x)}} - \overline{z+{d}_{(x)}}| < N/x^y$.
    Let $B$ be the smallest box that contains $p_{k}$ and $p_{k'}$, $L_1, \dots, L_D$ be the side lengths of $B$ and $v$ be its volume.
    Thus, $L_1 L_2 \dots L_D = v$.
    Let  $\ell_i$ be the integer such that 
    $ N / a_i^{\ell_i-1} > L_i \ge N / a_i^{\ell_i}$.
    Based on the above observation, $(k' - k)_{(a_i)}$ must contain $\ell_i-1$ leading zeros 
    or in other words, $a_i^{\ell_i-1}$ divides $k' - k$.
    Let $F = a_1^{\ell_1-1}a_2^{\ell_2-1}\dots  a_{D-1}^{\ell_{D-1}-1}$.
    Since $a_i$'s are relatively prime, it follows that $F$ also divides $k' - k$.
    However, observe that
    \[
      F = a_1^{\ell_1}a_2^{\ell_2}\dots  a_{D-1}^{\ell_{D-1}} \ge \frac{N^{D-1}}{L_1 L_2 \dots L_{D-1}} = \frac{N^{D-1}L_D}{v} = \frac{N^{D-1}(k'-k)}{v}.
    \]
    Let $X = \prod_{i=1}^{D-1} a_i$. 
    We claim $v \ge N^{D-1}/(2DX)$, since otherwise we will reach a contradiction.
    To see this, assume to the contrary that 
    $v < N^{D-1}/(2DX)$. 
    Assuming this, we get
    \begin{align}
      F = a_1^{\ell_1-1}a_2^{\ell_2-1}\dots  a_{D-1}^{\ell_{D-1}-1} \ge \frac{N^{D-1}(k'-k)}{ Xv} > k' - k \label{eq:orth}
    \end{align}
    which is a contradiction since $F$ must divide $k'-k$.

    It remains to estimate $X$.
    By Prime number theorem, we know that $a_i = O(i \log i)$.
    Thus, using Stirling's approximation, we have
    \[ 
       X =  \prod_{i=1}^{D-1} O(i\log i) \le  2^{O(D)}\cdot  D!\cdot (\log D)^D = D^{D + o(D)}.
    \]
\end{proof}

Using the above theorem, we prove the following result.

\begin{lemma}\label{lem:orthc}
  Consider the unit cube $\Q$ in $\R^D$.
  Let $N$ and $n_c$ be two integral parameters such that $n_c < N/2$.
  We can place a set $W$ of  $N$ points in $\Q$ and assign each an integer color from 0 to $n_c$ 
  such that the following hold:
  (i) for any two points $p$ and $q$ we have
   $\vol_D(\Box_D(p,q)) = \Omega(\frac 1{ND^{D+o(D)}})$ and
  (ii) for any two points $p$ and $q$ that have the same color we have
  $\vol_D(\Box_D(p,q)) = \Omega(\frac {n_c }{ND^{D+o(D)}})$.
\end{lemma}
\begin{proof}
    Consider the unit cube $U'$ in $R^{D+1}$.
    We use Theorem~\ref{thm:points} (after re-scaling the cube $[N]$ to the unit cube), 
    and we place a set $W'$ of $N$ points in $U'$.

    We now define the set $W$.
    Consider a point $p' \in W'$ and let $x$ be the value of the last coordinate of $p'$. 
    We project $p'$ into the first $D$-dimensions to get a point $p$ and color
    it with color $i = \lfloor xn_c \rfloor$. 

    We show that the projected points satisfy the two claims in the lemma. 
    Claim (i) is trivial:
    By Theorem~\ref{thm:points}, for any two points $p,q \in W$ that were obtained from the
    two points $p', q' \in W'$ we have
    $\vol_D(\Box_{D+1}(p',q')) = \Omega(\frac 1{ND^{D+o(D)}})$.
    Simply observe that 
    \[
      \vol_{D}(\Box_D(p,q)) \ge \vol_{D+1}(\Box_{D+1}(p', q')) = \Omega(\frac 1{ND^{D+o(D)}}).
    \]
    
    Thus, it remains to prove claim (ii).
    Consider two points $p, q \in W$ with the same color that correspond to 
    two points $p', q' \in W'$. 
    Since $p$ and $q$ have the same color, it follows that the difference between the value of their
    $D$-th coordinate is at most $1/n_c$.
    This fact combined by Theorem~\ref{thm:points} implies
    \[
      \vol_{D}(\Box_D(p,q)) \ge \vol_{D+1}(\Box_{D+1}(p', q'))n_c = \Omega(\frac {n_c}{ND^{D+o(D)}}).
    \]
\end{proof}

\subsubsection{Choosing the Colors.} \label{sec:colors}
In the previous subsection, we discussed constructions that will help us place the parametric points.
Here, we will pick the set of colors that are used to color them.
First, we establish an invariant. 
\paragraph{Invariant (I).}
Let $X:= \frac{\log_R \tau}{t}$.
We will maintain one invariant that $ w^{(j)}_i \ge 0$ and $w^{(j)}_i < X$, for each $1 \le j\le t-1$.
This invariant is to make sure that our construction is well-defined, in particular, to make sure  
that for each slab $\bs_i \in \S$, $\tau(\ts{j}_i)$, for all $1 \le j \le t$, are  in the
valid  range $(0,1]$.
As a result, any tuple of $t-1$ integers that satisfy this invariant, will yield well-defined
values for the thickness of the slabs used in our construction.

We will first need to estimate the total number of different colors that satisfy this invariant.
Let $\C$ be the set of all the colors satisfying Invariant (I).
In other words, $\C$ is the set of all $t-1$ tuples $(w_1, \dots, w_{t-1})$ where 
each $w_j$, $1\le j \le t-1$, is a non-negative integer and furthermore, 
$0 < R^{-w^{(j)}}\le 1$, for $1 \le j \le t$ where
$w^{(t)} = \log_R(\tau^{-1}) - \sum_{j=1}^{t-1}w^{(j)}$.

\begin{observation}\label{ob:colornum}
    $|\C| \ge X^{t-1}$.
\end{observation}
\begin{proof}
  If we force $0 \le w_j < X$, for $1 \le j \le t-1$, then we will have
  $w^{(t)} = \log_R(\tau^{-1}) - \sum_{j=1}^{t-1}w^{(j)} > 0$
  and thus $0 < R^{-w^{(j)}}\le 1$, for $1 \le j \le t$. Clearly, the number of tuples
  is at least as claimed.
\end{proof}

\paragraph{Pruning the colors.}
Fix an integral parameter $\ell$.
We call a subset $C\subset \C$ of $\ell$ colors an  $\ell$-subset. 
We say an $\ell$-subset $C$ is {\em bad} if 
by looking at the dimensions of the colors in $C$, we see only $2$ distinct values
at each dimension and $C$ is {\em good} if it is not bad. 
Alternatively, $C$ is good if we can find three colors $c_1, c_2, c_3 \in C$,
$c_1 = (w^{(1)}_1, w^{(2)}_1, \dots, w^{(t-1)}_1)$,
$c_2 = (w^{(1)}_2, w^{(2)}_2, \dots, w^{(t-1)}_2)$, and
$c_3 = (w^{(1)}_3, w^{(2)}_3, \dots, w^{(t-1)}_3)$,
and an index $j$ such that $w^{(j)}_1, w^{(j)}_2$, and $w^{(j)}_3$ are all distinct. 
Let $\C_g$ be the largest subset of $\C$ that contains no bad $\ell$-subsets
(in other words, every $\ell$-subset of $\C_g$ is good).

\begin{lemma}\label{lem:easygood}
    If $\ell > 2^{t-1}$, then $\C$ contains no bad $\ell$-subset and thus $\C_g = \C$.
\end{lemma}
\begin{proof}
    Consider a bad $\ell$-subset $C$.
    We can have at most $2$ distinct values at each coordinate of the tuples in $C$. 
    Therefore the number of tuples in $C$ cannot exceed $2^{t-1}$. 
    In turn, there are no bad $\ell$-subsets if $\ell>2^{t-1}$.
\end{proof}

\begin{lemma}\label{lem:hardgood}
    If $\ell \le 2^{t-1}$, then $|\C_g| = \Omega( X^{t-1}2^{-t}\Theta(X)^{-2t/\ell})$.
\end{lemma}
\begin{proof}
    We claim that there exists a subset $\C' \subset \C$ that contains the claimed number
    of colors without containing any bad $\ell$-subset and this clearly proves the lemma.

    We prove the claim using random sampling: we take a random sample of small enough size
    and then remove the bad subsets.

    Let $p$ be a parameter to be determined later. 
    Let $\C'$ be a subset of $\C$ where each color is sampled independently and with probability $p$. 
    Clearly, we have $\E(|\C'|) = p |\C| = p X^{t-1}$.
    From each bad $\ell$-subset, we remove one color.
    The set of remaining colors will be the claimed set $\C'$.
    By construction, $\C'$ will not contain any bad $\ell$-subsets but the main point is to
    show that $\C'$ will actually retain a significant fraction of the colors. 

    Let $\allbad$ be the total number of bad $\ell$-subsets $C \in \C$.
    We first estimate $\allbad$.
    By definition of a bad $\ell$-subset, for every dimension
    we see only $2$ distinct values among tuples in $C$.
    Thus, ${X \choose 2}$ is the total number of ways we can choose these distinct values,
    at a particular dimension.
    After choosing the distinct values, for every tuple, every dimension has only $2$ possible choices.
    Thus we have,
    \[ 
        \allbad \le \left( {X \choose 2} 2^\ell \right)^{t-1}.
    \]

    Now, consider a bad $\ell$-subset $C \subset \C$. 
    Observe that $C$ survives in $\C'$ with probability $p^\ell$ and thus
    the expected number of colors that we will remove is at most $\allbad p^\ell$.
    If we can choose the parameter $p$ such that $\allbad p^\ell \le 1$,
    then we are expected to only remove one color and thus 
    the expected number of colors left in $\C'$ after the pruning step is at least $pX^{t-1}/2$.
    Thus, we need to pick a value $p$ such that 
    \begin{align*}
        \allbad \cdot p^\ell \le\left( {X \choose 2} 2^\ell \right)^{t-1} p^\ell & \le 1  \Longleftarrow\\
        \Theta\left( X \right)^{2t} 2^{t\ell}p^{\ell} &\le 1 \Longleftarrow \\
        p =  2^{-t}& \Theta\left(X \right)^{-\frac{2t}{\ell}}.
    \end{align*}
    Picking $p$ as above, implies that the number of points left in $\C'$ is at least
    \[
        X^{t-1} 2^{-t}\Theta\left( X \right)^{-\frac{2t}{\ell}}. 
    \]
\end{proof}

\subsubsection{The Final Construction.} 
\seclab{finalconstruction}
We use Lemmas~\ref{lem:easygood} and \ref{lem:hardgood} (depending on the value
of $\ell$), to pick the set $\C_g$ of colors. 
Then, we use  Lemma~\ref{lem:orthc}, where $D$ is set to $t$, $N$ is set to $r$
and $n_c$ is set to $|\C_g|$.
Thus, Lemma~\ref{lem:orthc} yields us a point set $W$. 
The coordinates of the $i$-th point $\phi_i$ in $W$ defines the parametric point of $\bs_i$ and the color
of $\phi_i$ defines the thickness of the two-dimensional slabs that create $\bs_i$.
Thus, the set $W$ completely defines the set $\S$ of $r$ $t$-slabs that we aimed to build.

The last challenge is to bound the volume of the intersection of these slabs.
We will do this in the remainder of this subsection. 
\begin{lemma}\label{lem:lint}
  Consider $\ell$ $t$-slabs $\bs_1, \dots, \bs_\ell \in \S$ where 
   the defining parameters of $\bs_i$ are
  $(\alpha^{(1)}_i, \alpha^{(2)}_i, \dots, \alpha^{(t)}_i)$ and 
  $(w^{(1)}_i, w^{(2)}_i, \dots, w^{(t)}_i)$. 
  Then $\vol(\bs_1 \cap \bs_2 \dots \cap \bs_\ell)$ is asymptotically upper bounded by 
  \[
    \prod_{i=1}^t \min_{j,j', j\not = j'}\left\{\frac{\tau(\bs^{(i)}_j)\tau(\bs^{(i)}_{j'})}{|\alpha^{(i)}_j -  \alpha^{(i)}_{j'}|}\right\}.
  \]
\end{lemma}
\begin{proof}
    Consider $\Q_i$ and observe that $\ts{i}_1, \ts{i}_2, \dots, \ts{i}_\ell$ are slabs in $\Q_i$.
    Clearly, the region $\ts{i}_1\cap \ts{i}_2\cap \dots\cap \ts{i}_\ell$ is contained inside every region
    $\ts{i}_j\cap \ts{i}_{j'}$, for $1 \le j < j' \le \ell$.
    Thus, 
    \[
        \vol_2(\ts{i}_1\cap \ts{i}_2\cap \dots\cap \ts{i}_\ell) \le \min_{1 \le j < j' \le \ell}\left\{ \vol_2(\ts{i}_j\cap \ts{i}_{j'}) \right\}.
    \]
    The lemma follows from Observation~\ref{ob:slabint} since we have
    $\vol_2(\ts{i}_j\cap \ts{i}_{j'}) = \frac{\tau(\ts{i}_j)\tau(\ts{i}_{j'}) }{|\ta{i}_j - \ta{i}_{j'}|}$.
\end{proof}

We now present the main result of this subsection.
We recall the claim made in Lemma~\ref{lem:construction}.
\lemconstruction*
\begin{proof}
    Observe that the volume of the intersection any number of $t$-slabs is invariant under
    translation. 
    Thus, it suffices to look at the intersection of $\bs_1, \bs_2, \dots$, and $\bs_\ell$. 
    Let $\phi_1, \dots, \phi_\ell$ be their (colored) parametric points.
    We consider a few cases.
    \paragraph{Case I.} 
    In this case, we assume that two of the parametric points have the same color.
    W.l.o.g, assume the points $\phi_1, \phi_{2}$ have the same color. 
    Using Lemma~\ref{lem:orthc} we know 
    that $\vol_{t}(\Box_{t}(\phi_1, \phi_2)) = \Omega(n_c /(rt^{t+o(t)}))$ where $n_c = |\C_g|$.
    We can now bound
    \begin{align*}
        \vol_{2t}(\bs_1 \cap \bs_2 \cap \dots \cap \bs_\ell) &\le \vol_{2t}(\bs_1 \cap \bs_2) = O\left(  \frac{\tau^2}{\vol_{t}(\Box_t(\phi_1,\phi_2))} \right)\\
        & \le O\left( \frac{\tau^2rt^{t+o(t)}}{n_c } \right). 
    \end{align*}

    \paragraph{Case II.} 
    In this case, all the $\ell$ colors associated to $\phi_1, \dots, \phi_\ell$ are distinct. 
    By construction of our set of colors, we know that the set of $\ell$ colors assigned to $\phi_1, \dots, \phi_\ell$ is 
    a good $\ell$-subset. 
    This means, w.l.o.g, we can find three parametric points 
    $\phi_1 = (\alpha^{(1)}_1, \dots, \alpha^{(t)}_1), \phi_2 =(\alpha^{(1)}_2,
    \dots, \alpha^{(t)}_2)$, and $\phi_3 =(\alpha^{(1)}_3, \dots, \alpha^{(t)}_3)$  with colors
    $(w^{(1)}_1, \dots, w^{(t-1)}_1), (w^{(1)}_2, \dots, w^{(t-1)}_2)$, and $(w^{(1)}_3, \dots, w^{(t-1)}_3)$ 
    such that there is an index $j<t$ where $w^{(j)}_1, w^{(j)}_2$, and $w^{(j)}_3$ are all distinct. 

    To simplify the notation, let us rename $w_1 = w^{(j)}_1$, $w_2 =
    w^{(j)}_2$, and $w_3 = w^{(j)}_3$, $\tau_1 =\tau( \bs^{(j)}_1)$, $\tau_2 = \tau(\bs^{(j)}_2)$,
    and $\tau_3 = \tau(\bs^{(j)}_3)$,
    and $\alpha_1 = \alpha^{(j)}_1$, $\alpha_2 = \alpha^{(j)}_2$, and $\alpha_3 = \alpha^{(j)}_3$.
    Remember that we have constructed the slabs such that
    $\tau(\ts{j}_i) = R^{-\tw{j}_i}$ and thus $\tau_1 = R^{-w_1},  \tau_2 = R^{-w_2},$ and $\tau_3 = R^{-w_3}$.
    W.l.o.g, we can assume  $\alpha_1$ and  $\alpha_2$ make the closest
    pair among the three values of $\alpha_1, \alpha_2$, and $\alpha_3$ and that $w_1 > w_2$.
    We use Lemma~\ref{lem:lint}:
    \begin{align}
      \vol_{2t}(\bs_1\cap\bs_2\cap\bs_3) \le \prod_{i=1}^t \min\left\{\frac{\tau(\bs^{(i)}_1) \tau(\bs^{(i)}_{2})}{|\alpha^{(i)}_1 -  \alpha^{(i)}_{2}|},\frac{\tau(\bs^{(i)}_1 )\tau(\bs^{(i)}_{3})}{|\alpha^{(i)}_1 -  \alpha^{(i)}_{3}|},\frac{\tau(\bs^{(i)}_2) \tau(\bs^{(i)}_{3})}{|\alpha^{(i)}_2 -  \alpha^{(i)}_{3}|}\right\}.\label{eq:2dmath}
  \end{align}
  To upper bound the above, we consider two further cases. 
  \paragraph{Case II.A.}
    In this case, we assume  $w_{3} > w_{2}$.
    In the right hand side of Equation~\ref{eq:2dmath} above,
    for every index $i\not = j$, we pick the first term.
    In other words, we can write,
  \[
    \vol_{2t}(\bs_1\cap\bs_2\cap\bs_3) \le \min\left\{\frac{\tau_1 \tau_{2}}{|\alpha_1 -  \alpha_{2}|},\frac{\tau_1 \tau_{3}}{|\alpha_1 -  \alpha_{3}|},\frac{\tau_2 \tau_{3}}{|\alpha_2 -  \alpha_{3}|}\right\}\cdot \prod_{i=1, i\not = j}^t \frac{\tau(\bs^{(i)}_1) \tau(\bs^{(i)}_{2})}{|\alpha^{(i)}_1 -  \alpha^{(i)}_{2}|}.
  \]
  However, since $|\alpha_2 - \alpha_1| \le |\alpha_3 - \alpha_2|, |\alpha_3 - \alpha_1| $, we have:
  \begin{align*}
    \vol_{2t}(\bs_1\cap\bs_2\cap\bs_3) &\le \frac{\min\left\{\tau_1 \tau_{2},\tau_1 \tau_3,\tau_2 \tau_3\right\}}{|\alpha_1 -  \alpha_{2}|}\cdot \prod_{i=1, i\not = j}^t \frac{\tau(\bs^{(i)}_1) \tau(\bs^{(i)}_{2})}{|\alpha^{(i)}_1 -  \alpha^{(i)}_{2}|}  \\
    &=\min\left\{\tau_1 \tau_{2},\tau_1 \tau_3,\tau_2 \tau_3\right\}\cdot \frac{\prod_{i=1, i\not = j}^t \tau(\bs^{(i)}_1) \tau(\bs^{(i)}_{2})}{\vol_t(\Box_t(\phi_1,\phi_2))}  \\
    &=\min\left\{\tau_1 \tau_{2},\tau_1 \tau_3,\tau_2 \tau_3\right\}\cdot \frac{\frac{\tau}{\tau_1}\frac{\tau}{\tau_2}}{\vol_t(\Box_t(\phi_1,\phi_2))}.
  \end{align*}
  We now claim $\tau_1 \tau_3 \le \tau_1\tau_{2}/R$:
  remember that by our choice of index $j$, all the values $w_1, w_2$ and $w_3$  are distinct integers. 
  Thus, $w_3 > w_2$ implies $w_3 \ge w_2 - 1$ which in turn implies
  $\tau_3 \le \tau_2/R$ and thus $\tau_1 \tau_3 \le \tau_1\tau_{2}/R$.
  Using this and the fact that by Lemma~\ref{lem:orthc}, we have $\vol_t(\Box_t(\phi_1,\phi_2)) = \Omega(1/(rt^{t+o(t)}))$,
  we get that
  \begin{align*}
      \vol_{2t}(\bs_1\cap\bs_2\cap\bs_3)  \le \frac{r\tau^2 t^{t+o(t)}}{R}.
  \end{align*}

  \paragraph{Case II.B.}
  The remaining case is $w_{3} < w_{2}$ which combined with the assumption
  that $w_2 < w_1$ implies $w_3 < w_2 < w_1$.
  As before, since $w_1, w_2$, and $w_3$ are distinct integers,
  it follows that $\tau_1 \le \tau_2 / R$, and $\tau_2 \le \tau_3/R$.

  Because of the special geometry of the line, if $(\alpha_1, \alpha_2)$ is the closest pair, then it follows
  that $|\alpha_1 - \alpha_3| = \Theta(|\alpha_2-\alpha_3|)$.  

  In Eq.~\ref{eq:2dmath} and for every index $i \not = j$
  we pick the third term inside each minimization term. 
  In other words, we write,
  \begin{align*}
    \vol_{2t}(\bs_1\cap\bs_2\cap\bs_3) &\le \min\left\{\frac{\tau_1 \tau_{2}}{|\alpha_1 -  \alpha_{2}|},\frac{\tau_1 \tau_{3}}{|\alpha_1 -  \alpha_{3}|},\frac{\tau_2 \tau_{3}}{|\alpha_2 -  \alpha_{3}|}\right\}\cdot \prod_{i=1, i\not = j}^t \frac{\tau(\bs^{(i)}_2) \tau(\bs^{(i)}_{3})}{|\alpha^{(i)}_2 -  \alpha^{(i)}_{3}|} \\
    &\le \frac{\tau_1 \tau_{3}}{|\alpha_1 -  \alpha_{3}|}\cdot \prod_{i=1, i\not = j}^t \frac{\tau(\bs^{(i)}_2) \tau(\bs^{(i)}_{3})}{|\alpha^{(i)}_2 -  \alpha^{(i)}_{3}|}  \le \frac{\tau_1 \tau_{3}}{\Theta(|\alpha_2 -  \alpha_{3}|)}\cdot \prod_{i=1, i\not = j}^t \frac{\tau(\bs^{(i)}_2) \tau(\bs^{(i)}_{3})}{|\alpha^{(i)}_2 -  \alpha^{(i)}_{3}|} \\
    &= {\tau_1 \tau_{3}}\cdot \frac{\frac{\tau}{\tau_2}\frac{\tau}{\tau_3}}{\Theta(\vol_t(\Box_t(\phi_2,\phi_3)))} = \frac{\Theta(\tau^2)}{R \vol_t(\Box_t(\phi_1,\phi_2))} \le O\left( \frac{\tau^2 rt^{t+o(t)}}{R} \right).
  \end{align*}

  \paragraph{Putting it all together.}
  Combining all the above cases, we have shown that
    \[
    v=\vol_{2t}(\bs_1 \cap \bs_2 \cap \dots \cap \bs_m) \le  \max\left\{{\frac{\tau^2 rt^{t+o(t)}}{n_c}, \frac{\tau^2 rt^{t+o(t)}}{R}} \right\}.
    \]

   \end{proof}

\subsection{Multilevel Reporting Lower Bound}
\seclab{lb-reporting}
To prove a  lower bound for the multilevel reporting problem, we use the following theorem by Afshani.
We need the notion of a {\em geometric stabbing problem}:
the input is a set $\mathcal{R}$ of $n$ geometric
regions inside a $D$-dimensional region $\Q$ of volume 1.
An element of $\mathcal{R}$ is called a range.
The queries are points of $\Q$ and the output of a query $q$ is the subset of
ranges that contain $q$. 

\begin{theorem}\label{thm:framework}
	Assume we have a data structure for a geometric stabbing problem that uses
	at most $S(n)$ space and answers queries within $Q(n) + O(k)$ time in which $n$ is
	the input size and $k$ is the output size. 
	Assume for this problem we can construct an input set 
  $\mathcal{R}$ of $n$ ranges such that 
	(i) every point of $\Q$ is contained in at least $r$ ranges in which
	$r$ is a parameter greater than $Q(n)$ and (ii)
	the volume of the intersection of every $\alpha$ ranges is at most $v$, for
	two parameters $\alpha < r$ and $v$.
	Then, we must have $S(n) = \Omega(r v^{-1}/2^{O(\alpha)}) = 
	\Omega(Q(n) v^{-1}/2^{O(\alpha)})$.
\end{theorem}

\begin{theorem}\label{thm:2dlevel}
    Consider an algorithm $\A$ that given any set of $n$ $t$-slabs in $\R^2$, builds a pointer-machine
    data structure $\D$ of size $S(n)$ that solves the MSP in $Q(n) + O(k)$ time, where $k$ is the size of the
    output. In other words, given any $t$-point $\bp$, the data structure can output all the input $t$-slabs $\bs$
    that contain $\bp$ in $Q(n) + O(k)$.
    Then,
$S(n) = \Omega(( \frac{n}{Q(n)} )^2) \cdot \frac{\left( \frac{\log(n/Q(n))}{\log\log n}\right)^{t-1}}{2^{O(2^t)}}$
as well as
$S(n) = \Omega\left(  \frac{n}{Q(n)}\right)^2 { \Theta\left( \frac{\log(n/Q(n))}{t^{3+o(1)}\log\log n}\right)^{t-1-o(1)}}$.
\end{theorem}
\begin{proof} 
  We use Lemma~\ref{lem:construction}, where we set the parameter $r = Q(n)$, $\tau=2^{O(t)} r/n$, parameter $\ell$ to be determined, and
    $R=\left( \log n \right)^t$, which implies $X= \frac{\log_R
    (n/r)}{t} = \frac{\log(n/r)}{t^2 \log\log n}$.  
    The lemma gives us a set $\S$ containing $r$ $t$-slabs of thickness $\tau$.
    We now create a set $\mathcal{R}$ of $n$ $t$-ranges in the following way.
    For every $\bs_i \in \S$, we create a set $\S_i$ containing $O(n/r)$ disjoint translations of $\bs_i$
    such that the $t$-slabs in $\mathcal{R}_i$ cover $\Q$ entirely. 
    To be specific, consider a $t$-slab $\bs_i = (\ts{1}_i, \dots, \ts{t}_i)$.
    We tile $\Q_j$ using disjoint copies of $\ts{j}_i$ to obtain a set 
    $\S_{i,j}$ of two-dimensional slabs in $\Q_j$ (see Figure~\ref{fig:hyperslabs}).
    The set $\S_i$ is the Cartesian produce of these sets, that is,
    $\S_i = \S_{i,1} \times \dots \times \S_{i,t}$.
    Since each $\S_{i,j}$ tiles $\Q_j$, it follows that the set of slabs in $\S_i$
    tile $\Q$.
    Furthermore, the number of slabs in $\S_{i,j}$ is $\Theta(\frac{1}{\tau(\ts{j}_i)})$.
    This implies, the number of $t$-slabs in $\S_i$ is 
    \[
      |\S_i| \le \prod_{j=1}^t \Theta(\frac{1}{\tau(\ts{j}_i)}) = O\left( c^t \frac{1}{\tau(\bs_i)} \right) =O\left( c^t \frac{1}{\tau} \right).
    \]
    Since we have set $\tau = 2^{O(t)} r/n$, we can pick the constant in the exponent large enough 
    such that $|\S_i| \le n/r$.

    We let $\mathcal{R} = \S_1\cup \dots \cup \S_r$.
    By what we have just proved, $\mathcal{R}$ contains at most $n$ $t$-slabs.

\begin{figure}\centering
\includegraphics{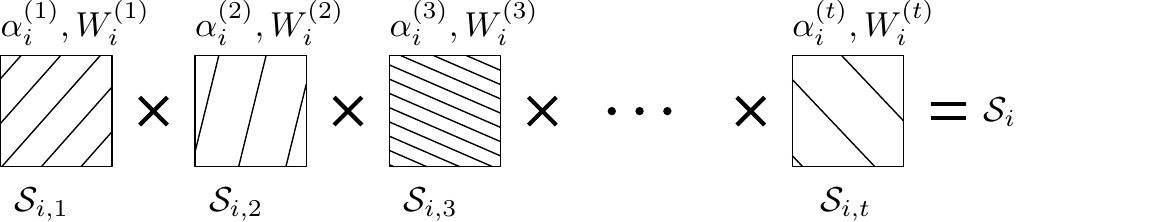}
\caption{We tile each $\Q_j$, then define the set of $t$-slabs in our construction 
    to be the Cartesian produce of these two-dimensional slabs.  }
  \label{fig:hyperslabs}
\end{figure}

    We look at the maximum volume, $v$, of the intersection of $\ell$ ranges $\bs_1, \dots, \bs_\ell$.
    If two of these $\ell$ $t$-slabs are the translations of the same $t$-slab from $\S$, then
    by construction, the volume of their intersection is empty. 
    Otherwise, we obtain a bound on $v$ by using Lemma~\ref{lem:construction}. 
    Since $n_c \leq X^{t-1}$, this implies $R\ge n_c$ and thus ${\frac{r^3}{n^2 n_c} \ge \frac{r^3}{Rn^2}}$
    and thus Lemma~\ref{lem:construction} gives us
    $v \le \frac{r^3 2^{O(t)} t^{t+o(t)}}{n^2 n_c} = \frac{r^3  t^{t+o(t)}}{n^2 n_c}$. 
    If we pick $\ell > 2^{t-1}$ (in particular, if we set $\ell = 2^t$) we have
    $n_c = X^{t-1} =\left( \frac{\log (n/r)}{t^2\log\log n}\right)^{t-1}$.
    Using the notation $a \gg b$ to denote $a = \Omega(b)$, 
    Theorem~\ref{thm:framework}, gives the following lower bound
    \begin{align*}
        S(n)&\gg Q(n) \cdot \frac{n^2  \Theta\left( \frac{\log(n/Q(n))}{t^2\log\log n}\right)^{t-1}}{Q(n)^3 t^{t+o(t)}}\cdot \frac{1}{2^{O(2^t)}}  \gg \left( \frac{n}{Q(n)} \right)^2 \cdot \frac{\left( \frac{\log(n/Q(n))}{\log\log n}\right)^{t-1}}{2^{O(2^t)}}.
    \end{align*}
    For small values of $t$ (e.g., constant $t$), this lower bound shows that the space/query time trade-off should
    increase by roughly a $\log n$ factor for every increase in $t$.
    However, 
    for larger values of $t$ the above lower bound degrades too quickly because of the
    $2^{O(2^t)}$ factor so we switch to the other branch in Lemma~\ref{lem:construction}.
   We set $\ell$ to be a value smaller than $2^t$ and
   obtain the bound $n_c =\Theta( X^{t-1}2^{-t}X^{-t/\ell})^{t-1}$.
   Note that we can again pick  $R=\left( \log n \right)^t$ which still satisfies $R \ge n_c$.
   Thus, we have $X:= \frac{\log_R (n/r)}{t} = \frac{\log(n/r)}{t^2\log\log n}$.
   In turn, we get that 
    \[
        n_c = \Theta( X^{t-1}2^{-t}X^{-2t/\ell}) = \Theta\left( \frac{\log(n/r)}{t^2\log\log n}\right)^{t-1}\cdot \left( \frac{\log(n/r)}{t^2\log\log n} \right)^{-2t/\ell}.
   \]

   This gives the lower bound
    \begin{align*}
        S(n)&\gg Q(n) \cdot \frac{n^2  \Theta\left( \frac{\log(n/Q(n))}{t^2\log\log n}\right)^{t-1}\cdot \left( \frac{\log(n/Q(n))}{t^2\log\log n} \right)^{-t/\ell}}{Q(n)^3 t^{t+o(t)}}\cdot \frac{1}{2^{O(\ell)}}  \\
        &\gg \left(  \frac{n}{Q(n)}\right)^2 \frac{ \Theta\left( \frac{\log(n/Q(n))}{t^{3+o(1)}\log\log n}\right)^{t-1}}{2^{O(\ell)}\cdot \left( \frac{\log(n/Q(n))}{t^2\log\log n} \right)^{t/\ell}}.
    \end{align*}
    We now can set $\ell= \Theta(\sqrt{t \log(\frac{\log n/Q(n)}{t})})$ to balance out the two terms in the denominator.
    This gives us the space lower bound of 
    \[
        S(n) \gg \left(  \frac{n}{Q(n)}\right)^2  \Theta\left( \frac{\log(n/Q(n))}{t^{3}\log\log n}\right)^{t-1-o(1)}.
    \]
\end{proof}

\subsection{The Lower Bound for \Frechet Queries in 2D}
\seclab{lower-bound-frechet}

We show how to use the construction from the previous section to prove the same
lower bound for Fr\'echet queries for polygonal curves in the plane.   
We first consider discrete \Frechet queries as a warm up, since they are
much easier to adapt our lower bound to. 

\subsubsection{Discrete \Frechet queries}\label{sec:lb_df}
The main idea is to simulate the phenomenon of a point stabbing a slab
using a point and intersection of two equal-sized circles. 
In particular, we use the following observation (see Figure~\ref{fig:lens-slab}).
\begin{figure}\centering
\includegraphics[scale=0.4]{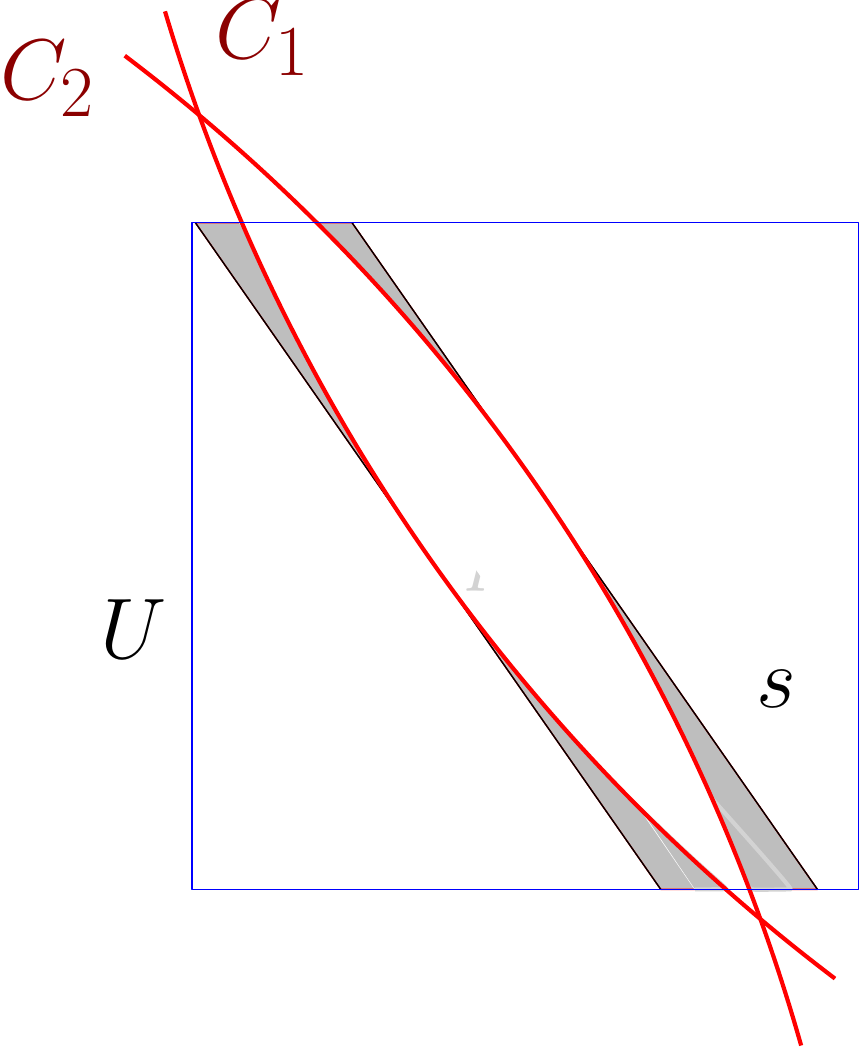}
\caption{
    We approximate a slab using the intersection of two circles.
    Given a slab $s$, we can find two circles $C_1$ and $C_2$ of equal radius such that
    their intersection is fully inside $s$ and when confined to the unit square
    $U$, their intersection almost covers the same area too.
    The symmetric difference of the slab and the intersection of the circles is shaded in grey.
    For any $\varepsilon > 0$, we can find the circles such that 
    the area of the grey region is less than $\varepsilon$.
}
\figlab{lens-slab}
\end{figure}

\begin{observation}\label{ob:lens}
    Given a slab $s$ and for any $\varepsilon$ we can find
    a value $\rho_{\varepsilon,s}$ such that for any value
    $\rho \ge \rho_\varepsilon$, we can place two circles $C_1$ and $C_2$ of  radius $\rho$ in the plane such that
    $C_1 \cap C_2 \subset s$ and the area of their error area, $U \cap s \setminus(C_1 \cap C_2)$, 
    is less than $\varepsilon$ where $U$ is the unit square.
\end{observation}

The idea is now very straightforward: we can replace the set of slabs used in our 
construction, with a set of lenses, i.e., the intersection of circles. 

Fix a global parameter $\varepsilon$.
Let $\mathcal{R}$ be the set of $t$-slabs used in the proof of Theorem~\ref{thm:2dlevel}.
To iterate, every $t$-slab $\bs \in \mathcal{R}$ is an element of the Cartesian product
of $t$ two-dimensional slabs $\ts{1}, \dots, \ts{t}$.
Using Observation~\ref{ob:lens}, we can approximate every slab $s=\ts{j}$ with a lens formed by the
intersection of two circles of
radius at least $\rho_{\varepsilon,s}$.
Let $\rho_\varepsilon$ be the maximum value of this radius, over all slabs $\ts{j}$ and over
all the $t$-slabs $\bs \in \mathcal{R}$.

We create $t$ unit square $\Q_1, \dots, \Q_t$ and place them such that the distance between them 
is greater than $10 \rho_\varepsilon$ (see Figure~\ref{fig:lb_discrete} on page \pageref{fig:lb_discrete}; the unit squares
are drawn closer in the figure for the purpose of illustration so one should imagine them far
enough that the circles intersecting a unit square $\Q_j$, do not intersect any other unit square.).
For every slab $\bs \in \mathcal{R}$ we create a chain $c(\bs)$ of size $2t$, by connecting the centers of the
circles that give rise to the lens that approximates $\ts{j}$.
A query $t$-point $\bp$ is simply represented by another chain $q(\bp)$ that connects the points $\tp{j}$,
$1 \le j \le t$.
See Figure~\ref{fig:lb_discrete}.

The only issue we are left with is that the lenses only approximate the slabs, meaning, there will
be $t$-points $\bp \in \Q$ that behave differently with respect to the slabs compared to the lenses.
Inside every unit square $\Q_j$, we have created $n$ lenses such that the error area of 
each lens is at most $\varepsilon$.
Thus, the error area of all the slabs created inside $\Q_j$ is at most $n\varepsilon$. 
Over all unit squares $\Q_j$, this error area is $nt\varepsilon$.
Remember that we had defined $\Q = \Q_1 \times \dots \times \Q_t$.
Define $\Q'$ as the subset of $\Q$ that includes all the $t$-points $\bp$ such that
none of the points $\tp{j}$ is inside an error area. 
By what we have observed, $\vol_{2t}(\Q') \ge 1 - nt\varepsilon$.
By picking $\varepsilon$ small enough, we can ensure that $\vol_{2t}(\Q') \ge 1/2$.

For a unit square $\Q_j$, we have placed all the centers of the circles that create
the slabs inside $\Q_j$, within distance of $\rho_\varepsilon$ of $\Q_j$.
Since we have placed the unit squares $\Q_j$ far apart, it means that a point of $q(\bp)$ inside 
$\Q_j$ can only be matched to the centers of the circles that create the slabs inside $\Q_j$.
Thus, a query $t$-point $\bp \in \Q'$ is inside a $t$-slab $\bs$ if and only 
the chain $q(\bp)$ is within discrete \Frechet distance $\rho_\varepsilon$ of the chain $c(\bs)$.

Observe that in the framework of Afshani (Theorem~\ref{thm:framework}), the region $\Q$ is the
set of all possible queries and it is only required to have volume one. 
To finish off, we rescale $\Q$ and all the slabs used in our construction by a constant factor
such that the volume of $\Q'$ equals one. 
Then, we apply the framework to the set of lenses (i.e., $t$-lenses) instead of $t$-slabs. 
Consider the requirement (i) in Theorem~\ref{thm:framework}. 
The construction in the previous section ensures that for every $t$-points $\bp=(\tp{1},\dots, \tp{t}) \in \Q'$, there are $r$ $t$-slabs
$\bs_1, \dots, \bs_r$ that contain $\bp$. 
Observe that this directly implies the existence of $r$ $t$-lenses that contain $\bp$ because $\tp{j}$ is not
contained in any error region, for all $1 \le j \le t$.
The requirement (ii) is trivially satisfied since lenses are created to be subsets of their corresponding
slabs, meaning, the volume of an intersection of lenses will have a smaller volume than the intersection
of their corresponding slabs. 

Thus, the lower bound of Theorem~\ref{thm:2dlevel} also applies to discrete \Frechet queries
where the input chains have complexity $2t$ and the query curves have complexity $t$.

\subsubsection{The continuous case}\label{sec:lb_dfc}
The construction in the previous subsection does not apply to the continuous \Frechet case.
The main problem here is that unlike the discrete case, it is not required for vertices
of the query chain to be mapped to the vertices of the input chain. 
As a result, an input chain $c(\bs)$ may match a query $q(\bp)$ even though 
the $t$-point $\bp$ is not contained in the $t$-slab $\bs$.

To resolve this issue, we
describe a construction of input curves that will take the role of the
$t-$slabs and we define a suitable set of query curves. 
Our construction does not vary the radius of the queries, we set the radius to $1$.

In the following, a polygonal curve is implicitly defined by a sequence of
vertices. To obtain the explicit curve, consecutive vertices need to be
linearly interpolated. The Cartesian product of two sets of polygonal curves simply 
concatenates the sequence of vertices thereby
effectively inserting the line segment that connects the endpoints of the
corresponding curves.

\paragraph{Zig-zag gadget}
Our input construction consists of concatenatenations of basic gadgets which we
call \emph{zig-zag gadgets} and which are described as follows.  The gadget is
constructed using parameters $x_1,x_2,x_3 \in [-1,1]$.  It is a polygonal curve
with four vertices $p_1,p_2,p_3,p_4$ defined as follows 
\[\pi(x_1,x_2,x_3) = \left( 
 (0,-4), 
 \left(x_1+\cos(\theta), \frac{x_2+x_3}{2} \right), 
 \left(x_1-\cos(\theta), \frac{x_2+x_3}{2} \right), 
 (0,4) 
 \right)\]
with $\theta$ defined by the equality $\sin(\theta)=\frac{|x_2-x_3|}{2}$.
An example is depicted in \figref{zigzag_gadget}.
Note that the two interior vertices $p_2$ and $p_3$ of the curve are chosen 
such that the two unit circles centered at $p_2$ and $p_3$ intersect on the 
vertical line segment from $(x_1,x_2)$ to $(x_1,x_3)$. 

\paragraph{Queries}
To simplify our analysis we will restrict the set of queries to polygonal
curves that have odd vertices on the vertical line at $-4$ and even vertices on
the vertical line at $4$.  We define the set of queries $\Q=\Q_1 \times \Q_2
\times \dots \Q_t$, where $\Q_i$ is the set of left-to-right line segments with
vertices $q_1=(-4,y_1)$ and $q_2=(4,y_2)$ for $y_1, y_2 \in [-1,1]$.
Each such query can be represented by an ordered set of lines 
$\ell_1,\dots,\ell_t$, such that $\ell_i$ is the line supporting the $i$th 
left-to-right edge of the query. Each line $\ell: y = a x + b$ can be 
represented as a point $p_{\ell}=(a,b)$ in the dual space of lines. We intend to use the
volume argument of Theorem \ref{thm:framework} in this dual space.

\begin{observation}
The set $\Q_i$ in the dual space forms a parallelogram of area
$1/2$.\footnote{Technically speaking, in order to obtain a set of queries with
area $1$ the space of queries needs to be scaled by a factor $2$. However, this
scaling does not affect our asymptotic bounds on the volumes.} This follows
from the fact that a line $\ell$ supports a line segment in $\Q_i$ if and only
if it intersects the two vertical intervals that define the set $\Q_i$.  In the
dual space this corresponds to the intersection of two slabs bounded by the
lines $y=4 x + 1$, $y=4 x - 1$, $y=-4 x +1$ and $y=-4 x -1$.
\end{observation}

\begin{figure}\centering
\includegraphics[width=\textwidth]{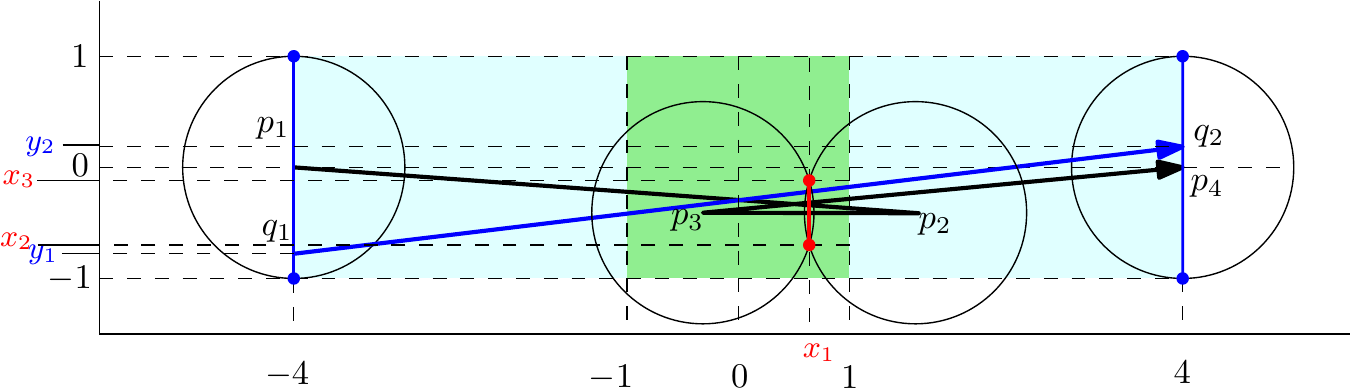}
\caption{Zig-zag gadget given by $p_1,p_2,p_3$ and $p_4$ (black curve) and query edge
given by $q_1$ and $q_2$ (blue edge). Any query edge that outputs the zigzag gadget
needs to intersect the vertical interval bounded by the two points $(x_1,x_2)$
and $(x_1,x_3)$. }
\figlab{zigzag_gadget}
\end{figure}

\begin{lemma}
For any zig-zag gadget $p=\pi(x_1,x_2,x_3)$ with $x_1,x_2,x_3 \in [-1,1]$ and
$|x_2-x_3|\leq 1$ and any line segment $q \in \Q_i$ we have $\distFr{p}{q} \leq
1$ if and only if $q$ is supported by a line $\ell$, whose dual $p_{\ell}$ lies
in the slab that has slope $x_1$ and intersects the $y$-axis in the interval
$[x_2,x_3]$.
\lemlab{dual-slabs}
\end{lemma}
\begin{proof} 
In the following, we refer to the intersection of the unit disks that are 
centered at $p_2$ and $p_3$ as the \emph{lens} and we denote with $\I$ 
the vertical interval formed by the two intersection points $(x_1,x_2)$ and $(x_1,x_3)$.
Guibas \etal~\cite{guibas1993approximating} proved that for a line segment $q$
it holds that $\distFr{p}{q} \leq 1$ if and only if $q$ stabs the unit disks
centered at the vertices of $p$ in the order along $p$. In the lemma by Guibas
\etal, the ordered stabbing requires a sequence of points on $q$ to exist, in
the order in which they appear on $q$, such that the $i$th point lies inside or
on the boundary of the $i$th disk centered at the vertices along the curve $p$. 
We claim that $q$ is an ordered stabber in this sense if and only if it
intersects the interval $\I$. Since $q \in \Q_i$, it must be that it intersects the disk at
$p_1$ and $p_4$ in the right order. This also implies that 
the slope of $q$ lies in the interval $[-\frac{1}{4}, \frac{1}{4}]$ and that
$q$ is directed from left to right. Therefore, $q$ needs to stab the disks at 
$p_2$ and $p_3$ in the lens.
Since we ensured $|x_2-x_3| \leq 1$, the slope of a line that stabs the lens
outside $\I$ is either larger or equal $\sqrt{3}$ or smaller or equal
$-\sqrt{3}$, thus the range of slopes of such lines is disjoint from the range
of slopes of query line segments in $\Q_i$.
Therefore, $q$ intersects $\I$ if and only if $\distFr{p}{q}\leq 1$. 
Now, the set of lines that intersects $\I$ corresponds to the set of points in the dual space
that lies in the intersection of two parallel halfspaces bounded by the lines
$y=x_1 x + x_2$ and $y= x_1 x + x_3$. This is the slab with slope $x_1$ which 
intersects the $y$-axis in the interval $[x_2,x_3]$.
\end{proof}

\paragraph{Input}
As in the previous section, we build $r$ different input sets $\S_1,\S_2,\dots,S_r$ 
such that:
(i) each $\S_i$ contains $\Theta(n/r)$ input curves 
(ii) for any two input curves  $\bs_1, \bs_2 \in \S_i$, the set of
queries that contains $\bs_1$ is disjoint from the set of queries that
contains $\bs_2$.
The set $\S_i$ is defined as the Cartesian product of the sets 
$\S_{i,1} \times S_{i,2} \times \dots S_{i,t}$, where each $\S_{i,j}$ is a set
of zig-zag gadgets defined by parameters $\alpha=\alpha^{(j)}_{i}$ and $W=W^{(j)}_{i}$; $W$ here will
play the role of the thickness.

We define a series of zig-zag gadgets with indices  
$1 \leq i' \leq \left\lceil \frac{2} {W} \right\rceil$:
$ \pi_{i'} = \pi(x^{i'}_1,x^{i'}_2,x^{i'}_3) $
with 
\begin{eqnarray*}
x^{i'}_1 &=& \tan(\alpha)\\
x^{i'}_2 &=& (i'-1) W - 1 \\
x^{i'}_3 &=& i' W - 1 
\end{eqnarray*}
The elements of this series form the set $\S_{i,j}$. 

Following \lemref{dual-slabs}, we call the slab with slope $x^{(i')}_1$ and
$y$-intercept $[x^{(i')}_2,x^{(i')}_3]$ the \emph{dual slab} of the
corresponding zig-zag gadget $\pi_{i'}$.  Furthermore, we call the Cartesian
product of $t$ dual slabs, each corresponding a zig-zag gadget of an element of
$\S_{i,j}$, a \emph{dual t-slab} of the corresponding element of $\S_i$.  Note
that \lemref{dual-slabs} puts some conditions on the parameters $\alpha$ and
$W$. In order to use the described representation in the dual space we need
that $-1 \leq x_1 \leq 1$, which translates to $-\frac{\pi}{4} \leq \alpha \leq
\frac{\pi}{4}$, and we need $|x_2-x_3| \leq 1$ which translates to $0 < W \leq
1$. However, this does not prevent us from using the lower bound construction
from the previous section as long as the range of angles for $\alpha$ is
constant. 

\begin{observation}
The angle of a dual slab of a zig-zag gadget in $\S_{i,j}$ is equal to 
$\alpha^{(j)}_i$ and the width of this slab is at most $W^{(j)}_i$. 
\end{observation}

\begin{observation}
The number of dual $t$-slabs in $\S_i$ is $\Theta(n/r)$.
\end{observation}

\begin{lemma}
For any two input curves  $\bs_1, \bs_2 \in \S_i$, the volume of the
intersection of their corresponding dual $t$-slabs is zero.
\end{lemma}

\begin{proof}
Note that the vertical intervals that define the zig-zag gadgets in $S_{i,j}$ tile the section of the 
vertical line at $x^{i'}_1$ which lies between the horizontal lines at $-1$ and $1$. 
It follows, by \lemref{dual-slabs}, that the zig-zag gadgets of $S_{i,j}$ tile
the set $\Q_i$ in the dual space using a series of pairwise disjoint slabs.
Therefore, the query curves of entire $\Q$ are partitioned by the $t$-slabs of
the set of query curves of $S_i$ such that the volume of the intersection of
any two dual $t$-slabs is zero.
\end{proof}

Thus, by combing the two constructions that we presented in this and the previous section,
we have proved the following theorem.

\begin{theorem}
Assume we have built a data structure for a  given set $\S$ of $n$ polygonal
curves of size $t$, and a fixed radius $\rho$, such that for any query 
polygonal chain $q$ of size $t$, we can find all the input curves within
the continuous or discrete \Frechet distance $\rho$ of $q$, in 
$Q(n) + O(k)$ time where  $k$ is the size of the output.

    Then,
$S(n) = \Omega(( \frac{n}{Q(n)} )^2) \cdot \frac{\left( \frac{\log(n/Q(n))}{\log\log n}\right)^{t-1}}{2^{O(2^t)}}$
as well as
$S(n) = \Omega\left(  \frac{n}{Q(n)}\right)^2 { \Theta\left( \frac{\log(n/Q(n))}{t^{3+o(1)}\log\log n}\right)^{t-1-o(1)}}$.

\end{theorem}

\section{A Data Structure}
\seclab{data-structure}
In this section, we will focus on building data structures to perform range
searching on polygonal curves based on the notion of \Frechet distance.  
Our data structures will ultimately use the recent results on semialgebraic range
searching, however, getting to the point where we can do that is non-trivial,
specially for the continuous \Frechet queries. 

Our data structures have two components: one component that is based
on recent results on semialgebraic range searching and the second component
that focuses on the \Frechet distance and tries to break down the data structure
problem into sub-problems that are  instances of semialgebraic range searching. 
Among these, the first component is very standard and not particularly interesting 
for the expert reader. 
The second component is where our contributions lie. 

In the next subsection, we will briefly go over the standard existing techniques
in range searching, combine them with the new results on semialgebraic range searching
and show how we can build multilevel data structures that can  handle more complex
semialgebraic input and query objects. 
In the two subsequent chapters, we will consider the discrete and continuous \Frechet
queries. 

\subsection{Multi-level Semialgebraic Range Searching}
\seclab{ml-data-structure}
We will use the following recent result from semialgebraic range searching.
Before stating the theorem, we will quickly cover some of the related
definitions.  By $\R_D[x_1, \dots, x_d]$ we denote the set of all $d$-variate
polynomials of degree at most $D$ (on variables $x_1, \dots, x_d$).  For a
polynomial $h \in \R_D[x_1, \dots, x_d]$, we denote the set of zeros of $h$
with $Z(h)$. In other words, $Z(h) = \left\{ (x_1, \dots, x_d) \in \R^d \mid
h(x_1, \dots, x_d) = 0 \right\}$.  For a given set $P$ of points, we say $Z(h)$
crosses $P$ if $Z(h)$ intersects any connected subset of $\R^d$ that contains
$P$. 
A semialgebraic set is a subset of $\R^d$ that satisfies some, $n_1$,
number of polynomial inequality of some degree, $n_2$, and using logical 
operands $\wedge, \vee$, and $\neg$. 

\begin{theorem}[The Semialgebraic Partition Theorem] Let $P$ be a set of $n$
points in $\R^d$ and let $r$ be a parameter.  There exists a constant $K$ that
only depends on $d$ such that the following hold.
 
  We can find $d$ integers $r \le r_1, \dots, r_d \le r^K$ such that we can
  partition $P$ into subsets $P=P^* \bigcup_{i=1}^d\bigcup_{j=1}^{t_i}P_{ij}$
  in which $P^*$ contains at most $r^K$ points, each $P_{ij}$ contains at most
  $n/r_i$ points, $t_i = r^{O(1)}$, and crucially, for any $d$-variate
  polynomial $h\in \R_D[x_1, \dots, x_d]$ where $D$ is another constant $Z(h)$
  intersects at most $O(r_i^{1-1/d})$ of the subsets $P_{i,1}, \dots,
  P_{i,t_i}$.

  Furthermore, each subset $P_{ij}$ is contained in a semialgebraic set
  $\Delta_{ij}$ that is defined by at most $O(r^{O(1)})$ polynomial inequalities
  of degree $O(r^{O(1)})$.  For any $d$-variate polynomial $h\in \R_D[x_1,
  \dots, x_d]$ $Z(h)$ intersects at most $O(r_i^{1-1/d})$ of the subsets
  $\Delta_{i,1}, \dots, \Delta_{i,t_i}$.

\end{theorem}

Using the semialgebraic partition theorem, we can solve the following 
multilevel semialgebraic range searching problem. 
The input is a set $\P$ $n$ of $t$-points in $\R^d$. 
The query is tuple of $t$ semialgebraic sets $(\rr_1, \dots, \rr_t)$, where each
semialgebraic set is defined by a constant number of polynomial inequalities of constant degree, and the goal is to find
all the points $\bp=(\tp{1}, \dots, \tp{t}) \in \P$ such that
the point $\tp{i}$ is contained in $\rr_i$, for $1 \le i \le t$;
say that such a $t$-point $\bp$ is contained in the tuple $(\rr_1, \dots, \rr_t)$.

As mentioned above, using the semialgebraic partition tree, and using classical techniques,
we can prove the following theorem.
The proof is included for completeness and also because of the fact that the existing literatures do not
explicitly mention such a data structure.
\begin{theorem}\label{thm:SARS} 
  Let $\P$ be a set of $n$ $t$-points in $\R^d$.
  We can store $\P$ in a data structure using $O(n O(\log\log n)^{t-1})$  space
  that can answer the following queries. 
  Given a tuple of $t$ semialgebraic sets $(\rr_1, \dots, \rr_t)$ where
  each $\rr_i$ is a semialgebraic set determined by a constant number of 
  polynomial inequalities of constant degree, we can output 
  all the $t$-points $\bp$  such that
  $\bp$ is contained in $(\rr_1, \dots, \rr_t)$.
  The query time is $O( n^{1-1/d} \log^{O(t)}n + k)$ where $k$ is the size of
  the output.

\end{theorem}
The rest of this section is devoted to the proof of the above theorem. 
We start with the description of the data structure.

    \paragraph{The Data Structure.} 
    We will describe a data structure $\D(\P)$ that is a multilevel data structure based
    on the Semialgebraic Partition Theorem.  
    Let $P_1$ be the set of first points of all the $t$-points in $\P$. 
    We use the Semialgebraic Partition Theorem with $P$
    set to $P_1$ and with parameter $r$ set to $n^\varepsilon$ for small enough
    constant $\varepsilon$ to be determined later.  
    This partitions $P_1$ into
    subsets $P_1 = P^* \bigcup_{i=1}^d\bigcup_{j=1}^{t_i}P_{ij}$.  
    We call these subsets ``canonical sets''.  
    Let $\P_{ij}$ be the canonical set that contains
    $t$-points whose first point is in the set $P_{ij}$.  Let $\P'_{ij}$ be the set of 
    $(t-1)$-points obtained by removing the first point of every $t$-point in $\P_{ij}$.
    Note that if $t=1$, then $\P'_{ij}$ is an empty set.  
    For every subset $\P_{ij}$ we do two different kinds
    of recursion.  Our first recursion is to build $\D(\P_{ij})$.  Our second
    recursion is to build $\D(\P'_{ij})$.  Our recursion stops as soon as $\P$
    contains a constant number of points.

\paragraph{Space Analysis.} We first analyze the space complexity of the data
structure.  Let $\spc_k(n)$ be the space complexity of the data structure if it
is run on an input of $n$ $k$-points. 
Our goal is to estimate $\spc_{t}(n)$.  
We have \[ \spc_{k}(n) = |P^*|
+  \sum_{i=1}^d\sum_{j=1}^{t}\spc_{k}(|S_{ij}|) +
\sum_{i=1}^d\sum_{j=1}^{t}\spc_{k-1}(|S'_{ij}|).  \]

It is easy to see that $\spc_{1}(n) = O(n)$  since the second recursion step do
not happen if $t = 1$ and thus each point is only stored in one sub-problem.
We guess that $\spc_k(n)$ solves to $O(n O(\log\log n)^{k-1})$ and try to prove
this with induction.  Thus, we can re-write the recursion as \[ \spc_{k}(n) =
|P^*| +  \sum_{i=1}^d\sum_{j=1}^{t}\spc_{k}(|S_{ij}|) + O(n O(\log\log
n)^{k-2}).  \]

Note that by the choice of $r$, each set $\P_{ij}$ has size at most $n/r_i \le
n/r = n^{1-\varepsilon}$.  Thus, there are $O(\log\log n)$ levels of recursion.
Observe that at each recursion level we are dealing with disjoint set of
subproblems.  This means that $S_k(n) =  O(n O(\log\log n)^{k-1})$.

\paragraph{The Query Algorithm.} 
Consider a query tuple $\br$ of $t$ semialgebraic sets $\rr_1, \dots, \rr_t$.
Let $\P_\br$ be the set of $t$-points in $\P$ that satisfy the query, i.e., $\P_\br$ contains
all the points $\bp=(\tp{1}, \dots, \tp{t}) \in \P$ such that
the point $\tp{i}$ is contained in $\rr_i$, for $1 \le i \le t$. 
Let $\P_{\br,1}$ be the set of $t$-points in $\P$ such that
only $\tp{1}$ is contained in $\rr_1$.
Clearly, $\P_\br \subset \P_{\br,1}$.
Our goal is to find the set $\P_{\br,1}$ as the disjoint union of a number
of canonical sets, that is, 
to find a set $C_{\rr_1}$ of canonical sets such that $\P_{\br,1} =
\bigcup_{c\in C_{\rr_1}}c$.  
We use the data structure $\D(\P)$.
Remember that in the data structure
$\D(\P)$ we have partitioned $P_1$ (the set of first points of $\P$) into subsets $P_1 = P^*
\bigcup_{i=1}^d\bigcup_{j=1}^{t_i}P_{ij}$.  
We explicitly process the $t$-points by looking at all the points in $P^*$ 
(that is, if we are solving a range reporting variant, we output them all, or if we are
solving a semigroup variant, we add up all the weights corresponding to $t$-points of $P^*$).
Next, we process the subsets $P_{i1}, P_{i2}, \dots, P_{it_i}$ starting from
$i=1$.  
By the Semialgebraic Partition Theorem, each $P_{ij}$ is contained in a
semialgebraic set $\Delta_{ij}$ and that each polynomial defining  the set $\rr_1$
intersects only $O(r_i^{1-1/d})$ of the sets $\Delta_{i1}, \Delta_{i2}, \dots,
\Delta_{it_i}$.  
Thus, the polynomial defining the set $\rr_1$ intersect at most $O(r_i^{1-1/d} )$ sets. 
We go through all the sets $\Delta_{i1}, \Delta_{i2}, \dots,
\Delta_{it_i}$ and for each $\Delta_{ij}$ determine (case a) if $\Delta_{ij}$
is completely outside $\rr_1$ (in which case we ignore it), or (case b)
$\Delta_{ij}$ is completely inside $\rr_1$ (in which case we add $P_{ij}$ as a canonical set
to $C_{\delta'}$) or (case c) if $\Delta_{ij}$ intersects
the boundary of $\rr_1$ and in this case we recurse on the data structure
$\D(\P_{ij})$.  Since each $\Delta_{ij}$ is determined by $r^{O(1)}$ polynomials
of degree $r^{O(1)}$, and $t_i = r^{O(1)}$, these tests will take $r^{O(1)}$
time in total.  
By the end of the recursion, we will have the desired set $C_{\rr_1}$. 

We have following recursion to describe the number of canonical sets
$f(n)$ placed in the set $C_{\rr_1}$.  
\[ 
    f(n) = r^{O(1)} + \sum_{i=1}^d O(r_i^{1-1/d} ) f(n/r_i).  
\] 
Note that we have $r=n^\varepsilon$ and that $r_i \ge r$.  
This is a standard recursion in the range searching
area and it is not too difficult to see that is solves to $f(n) =
O(n^{1-1/d}\log^{O(1)}n)$.

Having computed an implicit representation of $C_{\rr_1}$, we do the following.  
Remember that for every canonical set $c \in C_{\rr_1}$, we have
build another data structure $\D(c')$ where $c'$ is the set of $(t-1)$-points obtained
by removing the first point of the $t$-points in $c$. 
Any $t$-point $\bp$ represented by the canonical sets in $C_{\rr_1}$ has the property
that the point $\tp{1}$ is contained in $\rr_1$.
Thus, it remains to narrow the search such that $\tp{i}$ is also contained in $\rr_i$ for
$2 \le i \le t$. 
However, this is exactly equivalent to searching for points $\bp'=(\tp{2}, \dots, \tp{t})$
using the $t-1$ query tuples $\rr_2, \dots, \rr_t$.
Thus, we can simply recurse on each $c'$ (using $\D(c')$) for every $c \in C_{\rr_1}$.
Let $f_k(n)$ be the total number of canonical sets obtained after having recursed on a set containing $n$ $k$-points. We have the following recursion.
\[ f_{k+1}(n) \le r^{O(1)} +\sum_{i=1}^d \sum_{j=1}^{t_i} f_{k}(|P_{ij}|)  +
\sum_{i=1}^d O(r_i^{1-1/d}) f_{k+1}(n/r_i).  \] We guess that $f_k(n) =
O(n^{1-1/d}\log^{Ck}n)$ for a constant $C$.  
Since the sets $P_{ij}$ form a
partition of the set $P$ and they contain at most  $n$ $t$-points in total, the
recursion simplifies to \[ f_{k+1}(n) \le r^{O(1)} + O(n^{1-1/d}\log^{Ck}n)  +
\sum_{i=1}^d O(r_i^{1-1/d}) f_{k+1}(n/r_i).  \] 
Using the standard analysis
from the range searching literature, it is not too difficult to show that by
picking $C$ large enough we get $f_{k+1} = O(n^{1-1/d}\log^{C(k+1)}n)$.

\subsection{Discrete Frechet Queries} \seclab{dfq} 
Let $S$ be a set of $n$
polygonal chains in $\R^d$ where each chain $s \in S$ contains at most $t_s$ vertices.
For simplicity, we can assume every chain contains exactly $t_s$ vertices (by adding
extra dummy vertices). 
Consider a query polygonal chain $q$ of size $t_q$ and a chain $s \in S$.  
Imagine we would like to determine if the discrete Frechet distance between $q$ and $s$ is at most $\rho$,
for some parameter $\rho$. This can be done using the so-called
\emph{free-space-matrix}, which can be described as follows.  Let $M_{q,s}$
be the $0-1$-matrix with $t_q$ rows and $t_s$ columns, where the entry
$m(i,j)$ at row $i$ and column $j$ of $M_{q,s}$ is 0 if the distance
between the $i$-th vertex of $q$ and the $j$-th vertex of $s$ is greater
than $\rho$ and 1 otherwise.  Testing if the \Frechet distance between $s$ and
$q$ is at most $\rho$ now amounts to testing if there exists an $xy$-monotone
path connecting $m(1,1)$ to $m(t_q,t_s)$ that passes through the 1 entries.
We treat each input chain $s$ as a $t_s$-point and build the data structure from the
previous subsection.

We now describe the query procedure. 
Let $q$ be the query chain of $t_q$ vertices.  
Consider spheres of radius $\rho$ centered on the vertices of $q$.  
Let $\A$ be the arrangement created by the spheres.  
It is easy to see that the complexity of $\A$ is
$O(t_q^{d+1})$ by just lifting them to halfspaces in $\R^{d+1}$.  
Now, consider a chain $s \in S$ and the corresponding free-space matrix $M=M_{q,s}$.  
Every column of $M$ corresponds to a region in the arrangement $\A$.  
In other words, there are at most $t_q^{O(d)}$ 0-1 vectors could possibly appear as a column
in matrix $M$.  This in turn implies that the total number of matrices that can
be the free-space matrix of some chain in $S$ is upper bounded by
$t_q^{O(n_i)}$.  Let $\M$ be the set of these matrices.  We can compute $\M$
easily in  $t_q^{O(n_i)}$ time. 

During the query time, we will go through the following stages.  First, we
generate the set of matrices in $\M$.  For every matrix $M \in \M$, we will
only output chains $s$ with $M_{q,s} = M$.  Clearly, this will output all the
valid chains since $\M$ contains all the possible valid free-space diagrams and
it will not produce any duplicates since $M_{q,s}$ is unique.  Note that this
blows up the query time by a $t_q^{O(t_s)}$ factor.  Thus, in the second stage,
we have a fixed matrix $M \in \M$ and we would like to output the set of chains
$s$ such that $M_{q,s} = M$.  To do that, we ``triangulate'' $\A$: 
we lift the arrangement of spheres into $d+1$ dimension and
triangulate the resulting arrangement of halfspaces, and then project back to
$\R^d$.
This corresponds to decomposing $\A$  into $O(t_q^{d+1})$ cells where each cell is
a semialgebraic set determined by a constant number of polynomials of degree two.
Let $\A'$ be the resulting triangulation.  
Consider a chain $s$ such that $M_{q,s} = M$ and
consider the $i$-th column $v_i$ of $M$.  The bit vector $v_i$ encodes exactly which
points of $q$ are within distance $r$ of the $i$-th vertex of $s$.  In other
words, the bit vector $v_i$ identifies a unique cell $\delta_i$ in the arrangement
$\A$ such that the $i$-th vertex of $s$ must be contained in that cell.  
We are now almost done. 
Ideally, we would like to issue one query $(\delta_1, \dots, \delta_{t_s})$ to find exactly
what we want. 
However, the semialgebraic set $\delta_i$ might not be made using a constant number of polynomial inequalities. 
So we simply switch to the triangulated arrangement $\A'$.
Let $\delta_{i,1}, \dots, \delta_{i,x_i}$ be the set of cells formed in $\A'$ from triangulating $\delta_i$
where each cell is a semialgebraic cell formed by a constant number of polynomial inequalities of
constant degree. 
We now form $\prod_{i=1}^{t_s}x_i$ queries by creating the Cartesian product of these cells, that is,
$\left\{\delta_{1,1}, \dots, \delta_{1,x_1}  \right\}\times
\left\{\delta_{i,2}, \dots, \delta_{2,x_2}  \right\}\times
\left\{\delta_{i,t_s}, \dots, \delta_{i,x_{t_s}}  \right\}$.

Putting all these together, we can bound the total query time with \[
O(n^{1-1/d}\log^{O(t_s)}n\cdot  t_q^{O(t_s)}\cdot t_q^{O(d)} ) =
O(n^{1-1/d}\log^{O(t_s)}n\cdot  t_q^{O(t_s)}) \] assuming $t_q = O(\log^{O(1)}
n)$.

\begin{theorem}
  Given a set $S$ of $n$ polygonal curves in $\R^d$ where each curve contains 
  $t_s$ vertices, we can store $S$ in a data structure of 
    $\O\pth{n (\log\log n)^{t_s -1}}$ size such that given a query
    polygonal chain of size $t_q$ and a parameter $\rho$, it can output
    all the input curves within discrete \Frechet distance of $\rho$ to the query
    in  $\O\pth{n^{1-1/d} \cdot \log^{O(t_s)} n \cdot t_q^{\O(d)}}$, 
    assuming $t_q=\log^{O(1)} n$.  
\end{theorem}

\subsection{Continuous Frechet Queries} \seclab{cfq} 

Let $S$ be a set of polygonal chains as before. We now consider the
range-searching problem for the continuous \Frechet distance.  Consider a query
polygonal chain $q$ of size $t_q$ and a chain $s \in S$.  Imagine we would like
to determine if the \Frechet distance between $q$ and $s$ is at most $\rho$, for
some parameter $\rho$.  This can be done using the so-called \emph{free-space
diagram} which is a continuous version of the free-space matrix.

\paragraph{Free-space diagram} We can interpret the polygonal chains $s$ and
$q$ as continous curves $s: [0,1] \rightarrow \Re^2$ and $q: [0,1] \rightarrow
\Re^2$ by linearly interpolating consecutive vertices of the chain. Consider
the parametric space $[0,1] \times [0,1]$ of the two curves. The vertices of
the curves partition this parametric space into rectangular cells, such that
each cell corresponds to the parametric space of two edges, one from each
curve. The \emph{free-space} is the subset of points $(x,y) \in [0,1]\times
[0,1]$ such that $\|s(x)-p(y)\| \leq \rho$.  The free-space within each cell can
be described as an ellipse clipped to the cell and is therefore convex.  Now,
testing if the \Frechet distance of the two curves is smaller or equal to $r$
amounts to testing if there exists a $(x,y)$-monotone path that starts at
$(0,0)$ and ends at $(1,1)$ and stays inside the free-space.  We call such a
path \emph{feasible}.

\subsubsection{High-level Predicates} We would like to encode reachability in
the free-space diagram combinatorially using a small set of predicates. This
will help us to build an efficient data structure for the range-reporting
problem.  We denote the vertices of $s$ with $s_1,\dots,s_{t_s}$ and the
vertices of $q$ with $q_1,\dots,q_{t_q}$.

\begin{enumerate} \renewcommand{\theenumi}{(\Pre{\arabic{enumi}})}
\renewcommand{\labelenumi}{\theenumi} 
\item \emph{(Endpoints (start))} This
predicate returns true if and only if $\|s_1-q_1\| \leq \rho$ \label{ep}

\item \emph{(Endpoints (end))} This predicate returns true if and only if
$\|s_{t_s}-q_{t_q}\| \leq \rho$ \label{ep2}

\item \emph{(Vertex-edge (horizontal))} Given an edge of $s$, $\overline{s_j
s_{j+1}}$, and a vertex $q_i$ of $q$, this predicate returns true iff there
exist a point $p \in \overline{s_j s_{j+1}}$, such that $\|p-q_i\| \leq \rho$.
\label{hvep}

\item \emph{(Vertex-edge (vertical))} Given an edge of $q$, $\overline{q_i
q_{i+1}}$, and a vertex $s_j$ of $s$, this predicate returns true iff there
exist a point $p \in \overline{q_i q_{i+1}}$, such that $\|p-s_j\| \leq \rho$.
\label{vvep}

\item \emph{(Monotonicity (horizontal))} Given two vertices of $s$, $s_j$ and
$s_k$ with $j<k$ and an edge of $q$, $\overline{q_i q_{i+1}}$, this predicate
returns true if there exist two points $p_1$ and $p_2$ on the {\em line\/} supporting
the directed edge, such that $p_1$ appears before $p_2$ on this line, and such
that $\|p_1 - s_j\| \leq \rho$ and $\|p_2-s_k\| \leq \rho$.  \label{hmp}

\item \emph{(Monotonicity (vertical))} Given two vertices of $q$, $q_i$ and
$q_k$ with $i<k$ and an directed edge of $s$, $\overline{s_j s_{j+1}}$, this
predicate returns true if there exist two points $p_1$ and $p_2$ on the {\em line\/}
supporting the directed edge, such that $p_1$ appears before $p_2$ on this
line, and such that $\|p_1 - q_i\| \leq \rho$ and $\|p_2-q_k\| \leq \rho$.
\label{vmp} \end{enumerate}
%
%

\begin{lemma} \lemlab{hlp-correct} Given the truth values of all predicates
\ref*{ep}-\ref*{vmp} of two curves $s$ and $q$ for a fixed value of~$\rho$, one can
determine if $\distFr{s}{q} \leq \rho$.  \end{lemma}

Before we prove \lemref{hlp-correct}, we introduce the notion of a valid
sequence of cells in the free-space diagram and the set of predicates that are
induced by such a sequence.  In the following, we denote with $C_{i,j}$ the
cell of the free-space diagram that corresponds to the edges
$\overline{q_{i}q_{i+1}}$ and $\overline{s_{j}s_{j+1}}$.
We call a sequence of
cells $\C=((i_1,j_1),(i_2,j_2),\dots,(i_k,j_k))$ \emph{valid} if $i_1=1, j_1=1,
i_k=t_q-1, j_k=t_s-1$ and if for any two consecutive cells $(i_m,j_m)$ and
$(i_{m+1}, j_{m+1})$ it holds that either $i_{m}=i_{m+1}$ and $j_{m+1}=j_{m}+1$
or $j_{m}=j_{m+1}$ and $i_{m+1}=i_{m}+1$. Note that a sequence is valid if
there exists a feasible path which passes through the sequence of cells in the
right order. At the same time, there exists a valid sequence of cells for any
such path.  Any valid sequence of cells $\C$ induces a set of predicates $\P$
as follows.  \begin{enumerate}[(i)] \item $\ref*{ep} \in \P$ and $\ref*{ep2} \in
\P$ \item $\ref*{hvep}_{(i,j)} \in \P$ iff $(i,j-1),(i,j) \in \C$ \item
$\ref*{vvep}_{(i,j)} \in \P$ iff $(i-1,j),(i,j) \in \C$ \item
$\ref*{hmp}_{(i,j,k)} \in \P$ iff $(i,j-1),(i,k) \in \C$ and $j<k$ \item
    $\ref*{vmp}_{(i,j,k)} \in \P$ iff $(i-1,j),(k,j) \in \C$ and $i<k$
\end{enumerate}

We say that a valid sequence of cells is \emph{feasible} if the conjunction of
its induced predicates is true.  We claim that any feasible path through the
free-space induces a feasible sequence of cells and vice versa.  Before we
prove this claim, we prove the following helper lemma. Note the subtle
difference to the definition of the monotonicity predicate.

\begin{lemma}\lemlab{hmp-realize} Let $\C$ be a feasible sequence of cells and
    consider a monotonicity predicate $P$ of the set of predicates $\P$ induced by
    $\C$.  Let $a_1$ and $a_2$ be the vertices and let $e$ be the directed edge
    associated with $P$.  There exist two points $p_1$ and $p_2$ on $e$, such that
    $p_1$ appears before $p_2$ on $e$, and such that $\|p_1 - a_1\| \leq \rho$ and
    $\|p_2-a_2\| \leq \rho$.  \end{lemma}

\begin{proof} Assume $P$ is a horizontal monotonicity predicate
    $\ref*{hmp}_{(i,j,k)}$ in $\P$ (the arguments for vertical monotonicity
    predicates are similar).  The predicate $P$ was added because of cells
    $C_{(i,j-1)}$ and $C_{(i,k)}$ being present in $\C$. Since $\C$ is valid, it
    must be that $C_{(i,j)}$ and $C_{(i,k-1)}$ are also present (possibly with
    $j=k-1$). Therefore, $\P$ also contains horizontal vertex-edge predicates
    $\ref*{hvep}_{(i,j)}$ and $\ref*{hvep}_{(i,k)}$.  If all three predicates are
    true, then we want to follow that there exist points $p_1$ and $p_2$ on the
    \emph{edge} (not just the supporting line) $\overline{q_i q_{i+1}}$ such that
    $\|p_1 - s_j\| \leq \rho$ and $\|p_2-s_k\| \leq \rho$.

    We consider two cases, based on whether the common intersection of the line
    $\ell$ supporting the edge $\overline{q_i q_{i+1}}$ and the two disks of
    radius $\rho$ centered at $s_j$ and $s_k$ is empty, or in other words if 
    $\disk{s_j}{\rho}\cap \disk{s_k}{\rho} \cap \overleftrightarrow{q_iq_{i+1}} = \emptyset$. 
    If the intersection is
    empty, then the line intersects the disks in two disjoint intervals.  In
    this case any pair of points $p_1 \in \overline{q_i q_{i+1}} \cap \disk{s_j}{r}$
    and $p_2 \in \overline{q_i q_{i+1}} \cap \disk{s_k}{r}$ appears on $\ell$ in the
    correct order. Two such points $p_1$ and $p_2$ must exist since both vertex-edge predicates are true. 

        If the common intersection is not empty, then we argue that the edge
        $\overline{q_i q_{i+1}}$ must intersect this common intersection.
        Indeed, since both vertex-edge predicates are true, the edge intersects
        both disks.  Since the edge is a connected set, it must also intersect
        the common intersection which lies in between the intersections of the
        line with the two disks.  Now, if the edge intersects the common
        intersection, then we can choose $p_1=p_2$ from this common
        intersection.  We can make a similar argument for each vertical
        monotonicity predicate.  

        This implies that the points $p_1$ and $p_2$ of the monotonicity predicates are
realizable on the corresponding edges as claimed.  \end{proof}

\begin{proof}[Proof of \lemref{hlp-correct}] We claim that any feasible path
through the free-space induces a feasible sequence of cells and vice versa.
Assume there exists a feasible path $\pi$ that passes thorough the sequence of
cells $\C$.  Consider the endpoint predicate $\ref*{ep}$ (and respectively
\ref*{ep2}). The existence of $\pi$ implies that $(0,0)$ (and respectively $(1,1)$)
lies inside the free-space, which is equivalent to this predicate being true.
Now, consider a horizontal vertex-edge predicate $\ref*{hvep}_{(i,j)}$ for
consecutive pair of cells $C_{(i,j-1)}$, $C_{(i,j)}$ in the sequence $\C$. The
path $\pi$ is a feasible path that passes through the cell boundary between
these two cells.  This implies that the there exists a point on the edge
$\overline{q_i q_{i+1}}$ which lies within distance $\rho$ to the vertex $s_j$.
This implies that the predicate is true.  A similar argument can be made for
each vertical vertex-edge predicate.

Next, we will discuss the monotonicity predicates.  Consider a subsequence of
cells of $\C$ that lies in a fixed row $i$ and consider the set of predicates
$\P' \subseteq \P$ that consists of horizontal monotonicity predicates
$\ref*{hmp}_{(i,j,k)}$ for fixed $i$.  Let $p_j,p_{j+1},\dots,p_{k}$ be the
sequence of points along $q$ that correspond to the vertical coordinates where
the path $\pi$ passes through the corresponding cell boundaries corresponding
to vertices $s_j,s_{j+1}\dots,s_{k}$. The sequence of points lies on the
directed line supporting the edge $\overline{q_i q_{i+1}}$ and the points
appear in their order along this line in the sequence due to the monotonicity
of $\pi$. Since $\pi$ is a feasible path it lies in the free-space and
therefore we have $\| p_{k'}-s_{k'} \|$ for every $j \leq k' \leq k$. This
implies that all predicates in $\P'$ are true.  We can make a similar argument
for the vertical monotonicity predicates $\ref*{vmp}_{(i,j,k)}$ for a fixed
    column $j$.  This shows that a feasible path $\pi$ that passes through the
    cells of $\C$ implies that the conjunction of induced predicates $\P$ is
    true.

It remains to show the other direction: Any feasible sequences of cells implies
the existence of a feasible path.  It is clear that the relationship between a
feasible path $\pi$ and the endpoint predicates as well as the vertex-edge
predicates, as described above, gives us the existence of a continuous (not
necessarily monotone) path $\pi$ that stays inside the free-space and connects
$(0,0)$ with $(1,1)$. We now have to argue that the monotonicity predicates
imply that there always exists such a path that is also $(x,y)$-monotone. 

Assume for the sake of contradiction that the conjunction of predicates in $\P$
is true, but there exists no feasible path through the sequence of cells $\C$.
In this case, it must be that either a horizontal passage or a vertical passage
is not possible. Concretely, in the first case, there must be two vertices
$s_j$ and $s_k$ and a directed edge $e=\overline{q_i q_{i+1}}$, such that there
exist no two points $p_1$ and $p_2$ on $e$, such that $p_1$ appears before
$p_2$ on $e$, and such that $\|p_1 - s_j\| \leq \rho$ and $\|p_2-s_k\| \leq \rho$.
However, $\ref*{hmp}_{i,j,k}$ is contained in $\P$ and by \lemref{hmp-realize}
two such points $p_1$ and $p_2$ must exist. We obtain a contradiction. In the
second case, the argument is similar.  Therefore, a feasible sequences of cells
implies a feasible path, as claimed.  \end{proof}

\begin{lemma}\lemlab{compute-feasible} Given a truth assignment to all
predicates of two curves $q$ and $s$, we can decide if there exists a feasible
sequence of cells in $O(t_s t_q (t_s+t_q))$ time without knowing $q$ or $s$.  \end{lemma}

\begin{proof} Let $C_{i,j}$ denote the cell in the free space diagram that
corresponds to the $i$th edge on $q$ (the $i$th row) and the $j$th edge on $s$
(the $jth$ column).  For the sake of this proof, we re-define the notion of a
valid sequence of cells. We keep the definition as before, except that we drop
the requirement that $i_k=t_q-1, j_k=t_s-1$. That is, a valid path may end at
any cell. Such a path is feasible if the conjunction of its induced predicates
is true, as before. We want to process these cells in the lexicographical
ordering of their indices $(i,j)$  and determine for each cell whether it is
reachable by such a feasible sequence of cells. Concretely, this happens for a
cell $C_{i,j}$ if and only if there exists a feasible sequence of cells that
ends with $(i,j)$. Furthermore, we are interested in the second-to-last step.
If there exists a feasible sequence that ends with the two pairs
$(i-1,j),(i,j)$ we say that $C_{i,j}$ is \emph{reachable from below}, and
similarly if there exists a feasible sequence that
ends with the two pairs $(i,j-1),(i,j)$ we say that $C_{i,j}$ is
\emph{reachable from the left}.
In addition, for each processed cell $C_{i,j}$, we maintain two indices: 
\begin{enumerate}[(i)]
\item if $C_{i,j}$ is reachable from the left, we maintain the
maximal column index $j' \leq j$ such that $C_{i,j'}$ is reachable from
below.  
\item if $C_{i,j}$ is reachable from the below, we maintain the
maximal row index $i' \leq j$ such that $C_{i',j}$ is reachable from the
left.
\end{enumerate} 
We call these indices the \emph{previous right turn} and the \emph{previous
left turn}. Intuitively, the indices describe the index of the cell where the
feasible path that reached the cell from below previously turned left, and
respectively where the feasible path that reached the cell from the left
previously turned right. 

We now describe the algorithm.
If $\ref*{ep} \wedge \ref*{ep2}$ evaluates to false, the algorithm returns false.
Otherwise, we mark the cell $C_{1,1}$ as reachable from below and reachable
from the left. 
Processing a cell $C_{i,j}$ is done by executing the following steps: 
\begin{compactitem}
\item If $C_{i-1,j}$ is reachable from the left and if $\ref*{hvep}_{i,j}$
evaluates to true, then we mark $C_{i,j}$ as reachable from below and we set
its previous left turn to $i-1$.  

\item  If $C_{i-1,j}$ is reachable from below, let $i'$ denote its previous
left turn. If $\ref*{hvep}_{i,j}$ evaluates to true and if $\ref*{vmp}_{i,j,i''}$
evaluates to true for $i'\leq i''\leq i-1$, then we mark  $C_{i,j}$ as
reachable from below and we set its previous left turn to $i'$.

\item If $C_{i,j-1}$ is reachable from below and if $\ref*{vvep}_{i,j}$
evaluates to true, then we mark $C_{i,j}$ as reachable from the left and we set
its previous left turn to $j-1$.  

\item  If $C_{i,j-1}$ is reachable from below, let $j'$ denote its previous
right turn. If $\ref*{vvep}_{i,j}$ evaluates to true and if $\ref*{hmp}_{i,j,j''}$
evaluates to true for $j'\leq j''\leq j-1$, then we mark  $C_{i,j}$ as
reachable from the left and we set its previous right turn to $j'$.

\item Finally, if in the above steps we marked $C_{i,j}$ both as reachable from
below and reachable from the left, we set the previous right turn to $j$ and
the previous left turn to $i$.
\end{compactitem}

Clearly, processing each cell takes time at most $O(t_q+t_s)$. In total we are
processing $O(t_q t_s)$ cells. Therefore the total running time is $O(t_q t_s
(t_q + t_s))$.  The correctness of the algorithm can be proven by induction on
the cells in their processing order.
\end{proof}

\subsubsection{Low-level Predicates}
\newcounter{llp-counter}

In the following, we describe a set of simpler predicates that help us build 
a data structure. Each group of low-level predicates will be used to
represent one high-level predicate.  Let $a_1$ be the vertex and let
$\overline{b_1 b_2}$ be the edge of a vertex-edge predicate 
\ref*{hvep} (respectively, \ref*{vvep}). 
We define the following three predicates.
\begin{compactenum}[(a)]
\item $\|a_1-b_1\| \leq \rho$.  \label{a1b1}
\item $\|a_1-b_2\| \leq \rho$.  \label{a1b2}
\item $a_1$ is contained in the rotated rectangle $R$ with side lengths 
$2r$ by $\|b_1-b_2\|$ which is contained in the Minkowski sum of the edge
with the disk of radius $\rho$. (Refer to \figref{llp} (left))
\label{rectangle}
\setcounter{llp-counter}{\value{enumi}}
\end{compactenum}

\begin{lemma}
Given the truth values of the predicates (\ref*{a1b1})-(\ref*{rectangle}) one
can determine the truth value of the predicate \ref*{hvep} (respectively,
\ref*{vvep}).\footnote{This property was also used in the data structure by de Berg \etal \cite{de2013fast}}
\end{lemma}

\begin{proof}
Consider a vertex-edge predicate with corresponding vertex $a_1$ and edge
$\overline{b_1 b_2}$ and assume we know the true values of predicates (a),(b)
and (c). We can determine the truth value of the high-level certificate as $a
\vee b \vee c$. Indeed, the union of the two disks and the rectangle is equal
to the Minkowski sum of the edge with a disk of radius $\rho$. This is exactly the
locus of values for $a_1$ for which the predicate should return true.
\end{proof}

\begin{figure}
\centering
\includegraphics{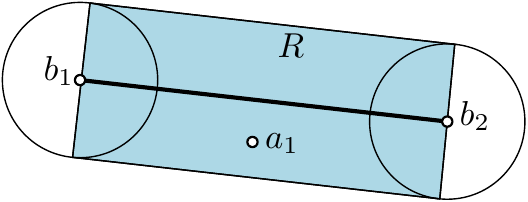}
\hspace{3em}
\includegraphics{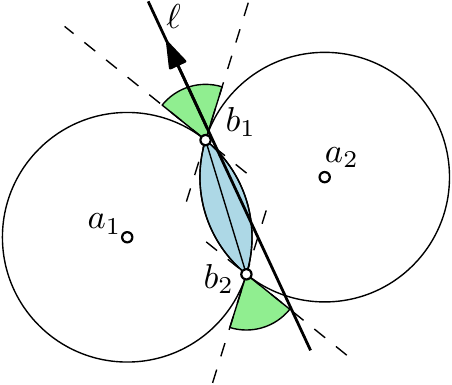}
\caption{Geometric objects involved in the low-level predicates.}
\figlab{llp}
\end{figure}

We also want to break down the monotonicity predicates into a constant number
of simpler predicates.  Let $a_1$, $a_2$ be the vertices and let $\ell$ be the
line supporting the directed edge $e$ of a monotonicity predicate \ref*{hmp}
(respectively, \ref*{vmp}). 
We begin with the following simple predicates
  
\begin{compactenum}[(a)]
\setcounter{enumi}{\value{llp-counter}}

\item The line $\ell$ intersects the circle of radius $\rho$ centered at $a_1$.
\label{line-circle-a1}
\label{firstitem2}

\item The line $\ell$ intersects the circle of radius $\rho$ centered at $a_2$.
\label{line-circle-a2}

\item The angle between the translation vector $(a_2-a_1)$ and the edge $e$ is
at most $\frac{\pi}{2}$.  
\setcounter{llp-counter}{\value{enumi}}
\label{line-direction}

\end{compactenum}
\medskip
In addition, we will distinguish the case that the two circles of radius $\rho$
centered at the two vertices intersect each other. This is captured by the
following predicate.
\begin{compactenum}[(a)]
\setcounter{enumi}{\value{llp-counter}}

\item $\|a_1-a_2\| \leq 2\rho$
\label{circle-circle}

\setcounter{llp-counter}{\value{enumi}}
\end{compactenum}

\begin{lemma}\lemlab{caseA}
If $\ref*{line-circle-a1} \wedge \ref*{line-circle-a2} \wedge
\ref*{line-direction}$ evaluates to true, then the corresponding monotonicity
predicate evaluates to true.  Furthermore, if \ref*{circle-circle} evaluates to
false or if the line does not intersect the lens formed by the two disks at $a_1$ and $a_2$, 
then $\ref*{line-circle-a1} \wedge \ref*{line-circle-a2} \wedge
\ref*{line-direction}$ is equivalent to the corresponding monotonicity predicate.
\end{lemma}

\begin{proof}
Consider a parametrized representation of the line given by a point 
$p_{\ell} \in \Re^2$ and a direction vector $v_{\ell} \in \Re^2$ with $\|v_{\ell}\|=1$
\[\ell: \{ p_{\ell} + t v_{\ell} ~|~ t \in \Re \}.\] 

Let $a_1' = p_{\ell} + \langle a_1-p_{\ell}, v_{\ell} \rangle v_\ell$ and let 
$a_2' = p_{\ell} + \langle a_2-p_{\ell}, v_{\ell} \rangle v_\ell$, where
$\langle\cdot,\cdot\rangle$ denotes the inner product.
Note that the point $a_1'$ (respectively $a_2'$) is an orthogonal projection
onto the line $\ell$ and minimizes the distance to $a_1$ (respectively, $a_2$).
Therefore (\ref*{line-circle-a1}) implies $\|a_1-a_1'\| \leq \rho$ and
(\ref*{line-circle-a2}) implies $\|a_2-a_2'\| \leq \rho$.  Furthermore,
(\ref*{line-direction}) implies that $a_1'$ appears before $a_2'$ on the line.
This implies that the corresponding monotonicity predicate evaluates to true.
Now, if (\ref*{circle-circle}) evaluates to false, then the line $\ell$ intersects
the two disks in disjoint intervals. In this case, the order along $\ell$ of any two 
points from these two intervals ($p_1$ in the intersection interval with the disk 
at $a_1$ and $p_2$ in the intersection interval with the disk at $a_2$)
indicates the correct truth value of the monotonicity predicate. Therefore also
$a_1'$ and $a_2'$ do. The same argument holds for the case that 
(\ref*{circle-circle}) evaluates to true and the line does not intersect the lens
formed by the two disks.
\end{proof}

In case (\ref*{circle-circle}) evaluates to true, let $b_1$ and $b_2$ be the two
intersection points of the two circles.  For this case we introduce the
following additional predicates.

\begin{compactenum}[(a)]
\setcounter{enumi}{\value{llp-counter}}

\item The line $\ell$ passes in between the two points $b_1$ and $b_2$ 
\label{line-b1-b2}

\item The angle of $\ell$ is contained in the range of angles of tangents of 
the circular arc between $b_1$ and $b_2$ of the circle of radius $\rho$ centered at $a_1$.
(Refer to \figref{llp} (right))
\label{arc-tangents-a1}

\label{lastitem2}
\end{compactenum}

\medskip
The purpose of the next lemma is to describe the cases in which the line
intersects the lens that is formed by the two disks at $a_1$ and $a_2$.
If (\ref*{circle-circle}) evaluates to false, the two points $b_1$ and $b_2$ are
undefined and we define the below predicates to be false. 

\begin{lemma}\lemlab{caseB}
If and only if $ \ref*{line-b1-b2} \vee \pth{ \ref*{line-circle-a1} \wedge
\ref*{line-circle-a2} \wedge \ref*{arc-tangents-a1}}$ evaluates to true, then the
line $\ell$ intersects the lens formed by the two disks of radius $\rho$ at $a_1$ and $a_2$.
\end{lemma}

\begin{proof} 
Clearly, if (\ref*{line-b1-b2}) evaluates to true, then the line intersects
the lens, since the line segment between $b_1$ and $b_2$ is contained in the
lens. For the remainder of the proof we intend to do a case analysis based on the angle 
of $\ell$ with respect to the truth value of (\ref*{arc-tangents-a1}). Refer to \figref{lens-projections}.
\begin{figure}
\includegraphics{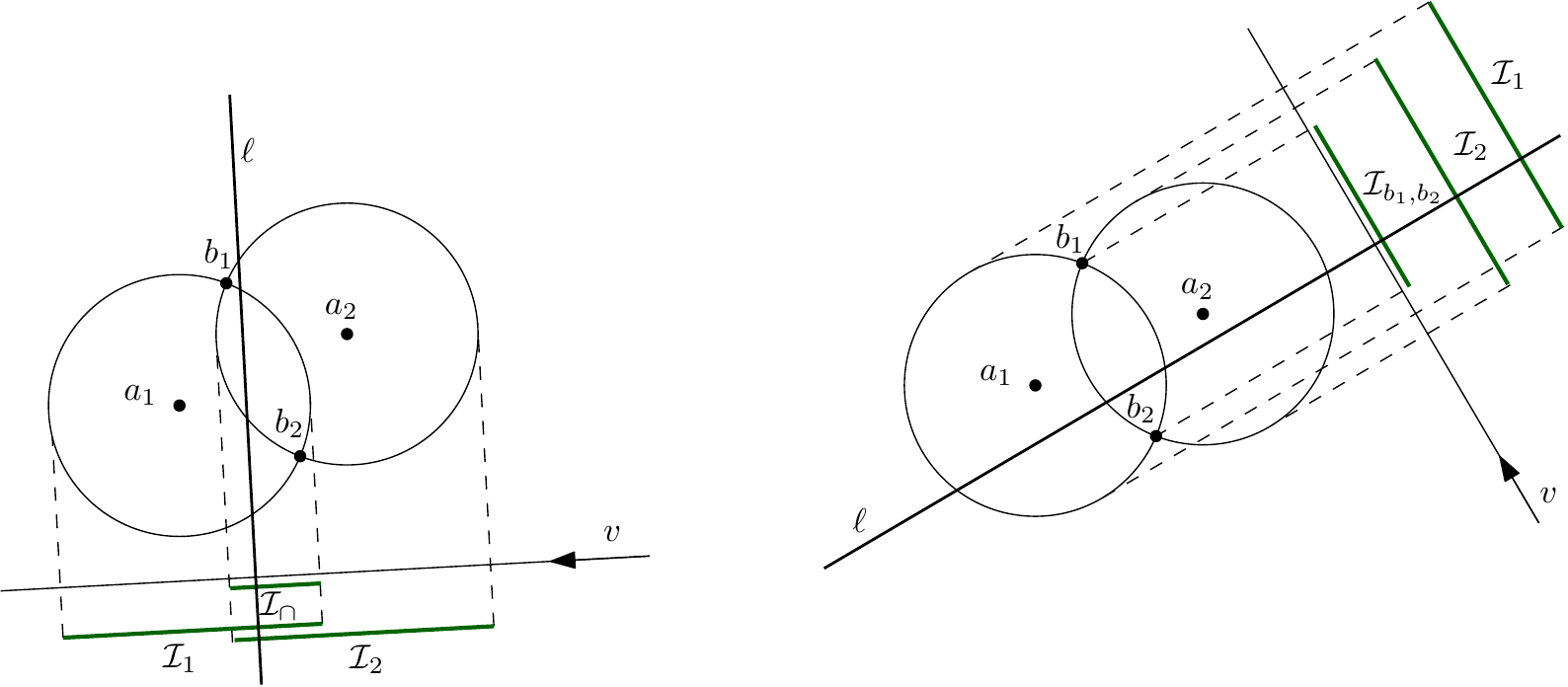}
\caption{Examples of the two cases analyzed in the proof of \lemref{caseB}.}
\figlab{lens-projections}
\end{figure}

In the following we denote with $\pi_{v}\pth{U}$ the projection 
of a set $U$ onto the subspace spanned by $v$:
\[\pi_{v}\pth{U} = \{\langle v,p \rangle ~|~ p \in U \}.\]

Let $v$ be the normal to the line $\ell$ and consider the projections of the
two disks onto this subspace 
$\I_1 = \pi_{v}\pth{\disk{a_1}{r}}$ and $\I_2 = \pi_{v}\pth{\disk{a_2}{r}}$ 
as well as the projection of the lens
$\I_{\cap}= \pi_{v}\pth{\disk{a_1}{r} \cap \disk{a_2}{r}}$

Note that, by symmetry of the lens, the range of angles of tangents is the same
on both arcs of the lens. Therefore, if (\ref*{arc-tangents-a1}) evaluates to true then 
the angle of $\ell$ lies in this range. 
In this case, we have that $\I_{\cap}=\I_1 \cap \I_2$, since the two bounding
tangents to the lens which are parallel to $\ell$ are also tangents---one to each of the two disks.
Therefore, the projection of any line $\ell'$ parallel to $\ell$ lies in the interval 
$\I_1 \cap \I_2$ if and only if it intersects the lens. This condition is equivalent to
$(\ref*{line-circle-a1} \wedge \ref*{line-circle-a2})$. 

Now assume that (\ref*{arc-tangents-a1}) evaluates to false. Consider the projection of 
the line segment bounded by $b_1$ and $b_2$ onto the same subspace as before
$ \I_{b_1,b_2} = \pi_{v}\pth{\overline{b_1 b_2}} $
For this range of angles (when (\ref*{arc-tangents-a1}) evaluates to false), we
have that $\I_{\cap} = \I_{b_1,b_2}$.  Therefore, in this case, the line
intersects the lens if and only if it intersects the line segment bounded by
$b_1$ and $b_2$. This condition is equivalent to (\ref*{line-b1-b2}).
\end{proof}

\begin{lemma}\lemlab{lens-predicate}
Assume (\ref*{circle-circle}) evaluates to true.  If the line $\ell$ intersects
the lens formed by the two disks at $a_1$ and $a_2$, then the corresponding
monotonicity predicate evaluates to true.
\end{lemma}

\begin{proof}
Let $p$ denote some point in the intersection of the line
with the lens. We have that $\|p-a_1\|\leq \rho$ and $\|p-a_2\|\leq \rho$, since
the intersection point lies inside the intersection of the two disks.  We can
set  $p_1=p$ and $p_2=p$ to satisfy the condition for the monotonicity
predicate to be true.
\end{proof}

\begin{lemma}
Given the truth values of the predicates (\ref*{firstitem2})-(\ref*{lastitem2}) one can determine the truth 
value of the predicate \ref*{hmp} (respectively, \ref*{vmp}).
\end{lemma}

\begin{proof}
We can determine the truth value of the corresponding monotonicity predicate as follows 
\begin{equation}\label{monotonicity-predicate}
\pth{\ref*{line-circle-a1} \wedge \ref*{line-circle-a2} \wedge \ref*{line-direction}}
\vee  \pth{\ref*{line-b1-b2} \vee \pth{ \ref*{line-circle-a1} \wedge
\ref*{line-circle-a2} \wedge \ref*{arc-tangents-a1}}}
\end{equation}

Indeed, together, \lemref{caseA}, \lemref{caseB} and \lemref{lens-predicate} testify that
Equation \eqref{monotonicity-predicate} implies the high-level predicate. 
In the other direction, we argue based on the case whether the line intersects
the lens. If the line does not intersect the lens, or the lens
does not exist (i.e., (\ref*{circle-circle}) evaluates to false), then by \lemref{caseA}, 
the predicate implies Equation \eqref{monotonicity-predicate}.
If the line does intersect the lens, then the predicate must be true and by \lemref{caseB},  
Equation \eqref{monotonicity-predicate} must be true as well.
\end{proof}


\subsubsection{The Data Structure}

We start by describing a set of binary matrices that will guide the layout of
the data structure and the query algorithm.  The matrices play a role similar
to the free-space matrix in \secref{dfq}.  The entries of these matrices
together define a truth assignment to the overall set of predicates.  The query
algorithm will process the matrices column by column.  We describe four
matrices, one for each group of high-level predicates. The predicates \ref*{ep}
and \ref*{ep2} are tested separately.

\begin{enumerate}
\item \emph{Horizontal vertex-edge predicates}\\
The matrix $\M_{\ref*{hvep}}$ consists of $t_q$ rows and $6 \cdot t_s$ columns.
We group the columns in $t_s$ groups of $6$ columns each.  The entries of the
$j$th group in the $i$th row are associated with the high-level predicate
$\ref*{hvep}_{(i,j)}$. For each $i$ and $j$, the predicate is a function of the
vertex $q_i$ of the query curve and the edge $\overline{s_js_{j+1}}$. 
Each entry stores the truth value of a low-level predicate associated with this
high-level predicate. We define these low-level predicates as follows.
Consider the rotated rectangle $R$ around $\overline{s_js_{j+1}}$ defined in
predicate (\ref*{rectangle}).  Denote with $\ell_1$ and $\ell_2$ the two
lines that bound $R$ from above and denote with  $\ell_3$ and $\ell_4$ the
two lines that bound $R$ from below.  A point $p$ is included in the $R$ if and
only if $p$ lies below $\ell_1$ and $\ell_2$ and $p$ lies above $\ell_3$ and
$\ell_4$.
The $6$ entries of the $j$th group of columns in row $i$ are defined as follows:
\begin{inparaenum}[(1)]
\item $s_j \in \disk{q_i}{r}$ 
\item $s_{j+1} \in \disk{q_i}{r}$ 
\item $\dual{\ell_1}$  lies below $\dual{q_i}$ 
\item $\dual{\ell_2}$  lies below $\dual{q_i}$ 
\item $\dual{\ell_3}$  lies above $\dual{q_i}$ 
\item $\dual{\ell_4}$  lies above $\dual{q_i}$. 
\end{inparaenum}

\item \emph{Vertical vertex-edge predicates}\\
Similar to the matrix $\M_{\ref*{hvep}}$, the matrix $\M_{\ref*{vvep}}$ consists
of $t_q$ rows and $6 \cdot t_s$ columns.
We group the columns in $t_s$ groups of $6$ columns each.  The entries of the
$j$th group in the $i$th row are associated with the high-level predicate $\ref*{vvep}_{(i,j)}$.
For each $i$ and $j$, the predicate $\ref*{vvep}_{(i,j)}$ is a function of the
vertex $s_j$ and the edge $\overline{q_iq_{i+1}}$ of the query curve. 
Consider the rotated rectangle $R$ around $\overline{q_iq_{i+1}}$ defined in
predicate (\ref*{rectangle}).
Denote with $\ell_1$ and $\ell_2$ be the two lines that bound $R$ from above and
denote with  $\ell_3$ and $\ell_4$ be the two lines that bound $R$ from below.
The $6$ entries are defined as follows:
\begin{inparaenum}[(1)]
\item $s_j \in \disk{q_i}{r}$ 
\item $s_j \in \disk{q_{i+1}}{r}$ 
\item $s_j$ lies below  $\ell_1$  
\item $s_j$ lies below $\ell_2$  
\item $s_j$ lies above $\ell_3$  
\item $s_j$ lies above $\ell_4$.  
\end{inparaenum}

\item \emph{Horizontal monotonicity predicates}\\
The matrix $\M_{\ref*{hmp}}$ consists of $t_q$ rows and $\frac{9}{2} t_s(t_s-1)$ columns.
We group the columns in $\frac{t_s(t_s-1)}{2}$ groups of $9$ columns each.
Each group is associated with a fixed value of $j$
and $k$, with $j<k$ and $j,k \in [t_s]$.  The entries of a specific group in
the $i$th row are associated with the high-level predicate $\ref*{hmp}_{(i,j,k)}$.  
This predicate is a function of the two vertices $s_j$ and $s_k$ and the edge
$\overline{q_i q_{i+1}}$ of the query curve.  Let $\alpha_{q} \in [0,2\pi]$ be the
angle of the translation vector $(q_{i+1}-q_i)$ with the $x$-axis and let
$\alpha_{s}$ be the angle of the translation vector $(s_k-s_j)$ with the
$x$-axis.  We denote with $\ell$ the line
that supports $\overline{q_i q_{i+1}}$. Let $\ell_{+}$ and $\ell_{-}$ be two
lines parallel to $\ell$ which lie at distance $\rho$ to $\ell$ and such that
$\ell_{+}$ lies above $\ell_{-}$.  If the circles of radius $\rho$ centered at
$s_j$ and $s_k$ intersect, then let $b_{+}$ and $b_{-}$ denote their
intersection points, such that $b_{+}$ has the larger $y$-coordinate of the two
points. Consider the range of angles of tangents to the intersection of the two
disks at $s_j$ and $s_k$. This range consists of two disjoint intervals of the
circular range of angles $[0,2\pi]$. At least for one of the two intervals the
left endpoint is smaller than the right endpoint (i.e., the interval does not
contain $2\pi$).  Let $[\alpha_{-},\alpha_{+}] \subseteq [0,2\pi]$ be this
interval.  Let $\alpha_{\ell} \in [0,2\pi]$ be one of the two angles of the
undirected line $\ell$ (we need to query with both of them).  The $9$ entries
are defined as follows:
\begin{inparaenum}[(1)]
\item $s_j$ lies below $\ell_{+}$
\item $s_j$ lies above $\ell_{-}$
\item $s_{k}$ lies below $\ell_{+}$
\item $s_{k}$ lies above $\ell_{-}$
\item $\alpha_s \in [\alpha_q-\frac{\pi}{2},\alpha_q+\frac{\pi}{2}]$ 
\item $b_{+}$ lies above $\ell$
\item $b_{-}$ lies below $\ell$ 
\item $\alpha_{-} \leq \alpha_{\ell}$ 
\item $\alpha_{+} \geq \alpha_{\ell}$. 
\end{inparaenum}

\item \emph{Vertical monotonicity predicates}\\
The matrix $\M_{\ref*{vmp}}$ consists of $\frac{t_q(t_q-1)}{2}$ rows and $8
\cdot t_s$ columns.  We group the columns in $t_s$ groups of $8$ columns each.
Each row is associated with a fixed value of $i$ and $k$, with $i<k$ and $i,k
\in [t_s]$. 
The entries of the $j$th group of columns in a specific row associated with
some value of $j$ and $k$ are associated with the high-level predicate
$\ref*{hmp}_{(i,j,k)}$.  This predicate is a function of the
two vertices $q_i$ and $q_k$ and the edge $\overline{s_j s_{j+1}}$.  
Let $\alpha_{q} \in [0,2\pi]$ be the
angle of the translation vector $(q_{k}-q_i)$ with the $x$-axis and let
$\alpha_{s}$ be the angle of the translation vector $(s_{j+1}-s_j)$ with the
$x$-axis. 
We denote with $\ell$ the line that supports $\overline{s_j s_{j+1}}$.
Let $\ell_{+}$ and $\ell_{-}$ be two lines parallel to $\ell$ which lie at
distance $\rho$ to $\ell$ and such that $\ell_{+}$ lies above $\ell_{-}$.  If the
circles of radius $\rho$ centered at $q_i$ and $q_k$ intersect, then let $b_{+}$
and $b_{-}$ denote their intersection points, such that $b_{+}$ has the larger
$y$-coordinate. Let $\I_q$ be the (not necessarily connected) range of angles of
tangents to the intersection of the two disks at $q_i$ and $q_k$. 
The $8$ entries are defined as follows:
\begin{inparaenum}[(1)]
\item $\dual{\ell_{+}}$ lies above $\dual{q_i}$ 
\item $\dual{\ell_{-}}$ lies below $\dual{q_i}$ 
\item $\dual{\ell_{+}}$ lies above $\dual{q_k}$ 
\item $\dual{\ell_{-}}$ lies below $\dual{q_k}$ 
\item $\alpha_s \in [\alpha_q-\frac{\pi}{2},\alpha_q+\frac{\pi}{2}]$ 
\item $\dual{\ell}$ lies below $\dual{b_{+}}$
\item $\dual{\ell}$ lies above $\dual{b_{-}}$
\item $\alpha_{s} \in \I_{q}$. 
\end{inparaenum}
\end{enumerate}

Following the lemmas in \secref{cfq} we can make the following observation.  

\begin{observation}\obslab{feasible-truth}
Given an instance of each matrix $\M_{\ref*{hvep}}, \M_{\ref*{vvep}},
\M_{\ref*{hmp}},$ and $\M_{\ref*{vmp}}$ and given the values of the two
predicates \ref*{ep} and \ref*{ep2}, one can determine the values of the
high-level predicates using \lemref{caseA}, \lemref{caseB} and
\lemref{lens-predicate}.
\end{observation}

\begin{observation}\obslab{colnum}
The total number of columns of the matrices $\M_{\ref*{hvep}}, \M_{\ref*{vvep}},
\M_{\ref*{hmp}},$ and $\M_{\ref*{vmp}}$ is bounded by $O(t_s^2)$.
\end{observation}

We are now ready to state our problem in the terms of a multilevel
semialgebraic range searching problem.
The input for the multilevel semialgebraic range searching problem is a set 
$\P$ of $n$ $t$-points in $\R^2$ which we define as follows. 
Each of the columns of the matrices $\M_{\ref*{hvep}}, \M_{\ref*{vvep}},
\M_{\ref*{hmp}},$ and $\M_{\ref*{vmp}}$ defines a point $p \in \Re^2$ (or a
value $p \in \Re$, in this case we set the second coordinate to 0) derived from
each input curve. Let $c$ be the total number of columns of the four
matrices\footnote{ Indeed, the ordering of the columns does not matter as long
as it is consistent throughout the data structure and query algorithm.}.
We define a $t$-point $(\tp{1},\tp{2}, \dots, \tp{t})$ for each input curve $s$
with $t=c+2$ such that $\tp{k}$ for $k \leq c$ is the point of $s$ that is
defined by the predicate in the $k$th column, where $k \in [1,c]$ 
uniquely identifies a column across the four matrices. For $k=c+1$ we define
$\tp{k}=s_1$ and for $k=c+2$ we defined $\tp{k}=s_{t_s}$.  Note that these two
point sets are associated with the predicates $\ref*{ep}$ and $\ref*{ep2}$,
which are not captured by the matrices.
We define the set of $t$-points obtained this way to be the set $\P$. 

Before we describe the queries, we note that using the data structure described
in \secref{ml-data-structure} we use space in $S(n) =  O(n O(\log\log
n)^{t-1})$, where $t$ is in $O(t_s^2)$ by \obsref{colnum}.

In the following, we describe the queries.  As before, each query is a tuple of
$t$ semialgebraic sets $(\rr_1, \dots, \rr_t)$, where each semialgebraic set is
defined by a constant number of polynomial inequalities of constant degree.

Each entry of the matrices $\M_{\ref*{hvep}}, \M_{\ref*{vvep}},
\M_{\ref*{hmp}},$ and $\M_{\ref*{vmp}}$ defines a query range derived from the
query curve that is either a disk or a halfspace. In the first phase of the
query algorithm we compute the arrangement of these ranges for each column.
Every cell of the partition corresponds to a truth assignment to this column.
In this way we generate all possible truth assignments that correspond to
non-empty query ranges. Each truth assignment to a column can be stored as a
bit vector.  The Cartesian product of the set of bit vectors generated this way
yields a set of truth assignments to the matrices $\M_{\ref*{hvep}},
\M_{\ref*{vvep}}, \M_{\ref*{hmp}},$ and $\M_{\ref*{vmp}}$. From this set we
want to use those truth assignments only that have a feasible sequence of
cells. We can test this for each generated matrix using \obsref{feasible-truth}
and \lemref{compute-feasible} in $O(t_s t_q (t_s+t_q))$ time.

In the second phase of the query algorithm we have a fixed truth assignment and
for each column we compute the cell of the arrangement corresponding to this
truth assignment. We now refine the cells as described in \secref{dfq},
by lifting circular ranges to $\Re^3$ and mapping back the edges of the
triangulation to $\Re^2$ in order to obtain cells that can be described by a 
constant number of polynomial inequalities of constant degree.
This generates a set of ranges for each column. Finally, we take the Cartesian
product of these ranges as done in \secref{dfq}.

Putting all of the above together, we can bound the total query time with \[
O(\sqrt{n}\log^{O(t_s^2)}n\cdot  t_q^{O(t_s^2)} ) = 
O(\sqrt{n}\log^{O(t_s^2)}n )
\] assuming $t_q = O(\log^{O(1)}
n)$.

\begin{theorem}
  Given a set $S$ of $n$ polygonal curves in $\R^d$ where each curve contains 
  $t_s$ vertices, we can store $S$ in a data structure of 
    $\O\pth{n (\log\log n)^{O(t_s^2)}}$ size such that given a query
    polygonal chain of size $t_q$ and a parameter $\rho$, it can output
    all the input curves within continuous \Frechet distance of $\rho$ to the query
    in  $\O\pth{\sqrt{n} \cdot \log^{O(t_s^2)} n }$, 
    assuming $t_q=\log^{O(1)} n$.  
\end{theorem}

\section{Conclusions}
\seclab{conclusions}

We studied the space/query-time tradeoff of multi-level data structures for
range searching under the \Frechet distance.  The aim of our study was
two-fold. On the one hand, we wanted to answer a fundamental question related
to range searching: Do multilevel data structures need to have an exponential
dependency on the number of levels? We answer the question to
the negative by proving a lower bound on the space/query time tradeoff for a
concrete problem which we refer to as multilevel stabbing problem. In
particular, our lower bound shows that finding a general technique for removing
the exponential dependency on the number of levels is not feasible.
On the other hand, we were interested in the complexity of range searching
among polygonal curves under the \Frechet distance. Previous to our work, the
complexity of this problem was largely open.  We give upper and lower bounds on
the space/query-time tradeoff for both the discrete and continuous versions of
the \Frechet distance.  The fact that we can extend our lower bound to such a
practically relevant problem further supports our claimed
negative answer to the broader range searching question mentioned above.
Our data structures invoke semialgebraic range searching within the framework
of multilevel partition trees.  Here, the major challenge lies in the \Frechet
distance not being defined as a closed form algebraic expression. In other
words, previous to our work, it was not at all obvious how semialgebraic
range searching could be applied to range searching under the \Frechet
distance. 
Our upper bounds for this problem are in line with the lower bounds, as the number of 
levels is in the order of $t$, the complexity of the polygonal curves. For the continuous
\Frechet distance, the number of levels increases to $O(t^2)$. We think that it
can be reduced to $O(t)$ by selectively using dualization in some of the
levels. 
However, doing so requires dealing with more technical details and 
given the technical nature of the current results, we have decided to
pursue this improvement as a part of future work. 

Finally, we can identify a number of interesting open problems that remain. We close our
discussion by stating some of them in no particular order.

\begin{compactenum}
\item Can the use of semialgebraic range searching be circumvented for \Frechet queries?
  A positive answer could help us solve the following open question.

\item Can we develop data structures for (discrete or continuous) \Frechet
queries that match the lower bounds in the high-space/low-query-time regime?

\item How fast can we do approximate range searching under the (discrete or
continuous) \Frechet distance? Using the standard $(1+\eps)$-approximation of
spherical ranges, we can use simplex-range searching to get the full spectrum
of the space/query time tradeoff, but even more efficient data structures might
be possible.

\item Can we prove lower bounds for answering approximate \Frechet distance 
  queries? One particular challenge here is that the known lower bound frameworks
  cannot handle approximations. For example, this might mean that we need to
  extend the pointer machine lower bound framework of Afshani~\cite{a12}. 

\item Can we do range searching under the continuous \Frechet distance among
polygonal curves in three or higher dimensions? 
In particular, we are interested in a data structure that uses $O(n \log^{t^{O(1)}}n)$
space  and answers queries in $O(n^{1-1/d}\log^{t^{O(1)}}n)$ time where 
$t$ is the maximum complexity of a query or input curve.
We suspect the answer might be negative.

\item Can we do range searching under related distance measures such as dynamic
time warping?

\item Do our lower bounds extend to multilevel stabbing queries in the semigroup model?

\item Can we prove stronger lower bounds for polygonal curves in three or higher
dimensions? We hope that our upper bounds for the discrete \Frechet distance
could be matched.

\end{compactenum}

For resolving the last two questions, the main challenge lies in proving the
equivalent of Theorem~\ref{thm:points} on page \pageref{thm:points} and a case
analysis such as the one in \secref{finalconstruction}.

\bibliographystyle{abbrv}
\bibliography{ref}

\end{document}